\documentclass{amsart}
\usepackage{mathrsfs}
\usepackage{amsaddr}
\usepackage[usenames, dvipsnames]{color}
\title[Integrable evolutions of twisted polygons]{Integrable evolutions of twisted polygons in centro-affine $\RR^m$}
\author{Gloria Mar\'i Beffa \& Annalisa Calini}
\address{Mathematics Department, University of Wisconsin, Madison, WI USA}
\address{Department of Mathematics, College of Charleston, Charleston, SC  USA}
\email{maribeff@math.wisc.edu, calinia@cofc.edu}
\usepackage{amsmath,amsfonts,amsthm,amssymb,relsize,mathdots,MnSymbol}

\usepackage[symbol]{footmisc}

\newtheorem{theorem}{Theorem}
\newtheorem{conj}[theorem]{Conjecture}
\numberwithin{theorem}{section}
\newtheorem{corollary}[theorem]{Corollary}
\newtheorem{cor}[theorem]{Corollary}
\newtheorem{lemma}[theorem]{Lemma}
\newtheorem{lem}[theorem]{Lemma}

\newtheorem{proposition}[theorem]{Proposition}
\newtheorem{prop}[theorem]{Proposition}  

\theoremstyle{definition}
\newtheorem{definition}[theorem]{Definition}
\newtheorem{example}[theorem]{Example}
\newtheorem*{remark}{Remark}

\def\RR{\mathbb R}
\def\RP{\mathbb{RP}}
\def\ZZ{\mathbb Z}
\def\NN{\mathbb N}
\def\SL{\mathrm {SL}}

\def\r{\mathbf{r}}
\def\gam{\gamma}
\def\ddt{\dfrac{\mathrm{d}}{\mathrm{d}t}}
\def\T{\mathcal{T}}
\def\V{\mathcal{V}}
\def\Rc{\mathcal{R}}

\def\M{\mathcal{M}}
\def\F{\mathcal{F}}
\def\H{\mathcal{H}}
\def\K{\mathcal{K}}

\def\Pc{\mathcal{P}}
\def\sl{\mathfrak{sl}} 
\def\e{\mathbf{e}}
\def\p{\mathbf{p}}
\def\g{\mathfrak{g}}
\def\h{\mathfrak{h}}
\def\k{\mathbf{k}}
\def\a{\mathbf{a}}

\def\0{\mathbf{0}}
\def\I{\mathbf{I}}
\def\X{{\widehat{X}}}
\def\Y{{\widehat{Y}}}
\def\Z{{\widehat{Z}}}
\def\Vh{\widehat{V}}

\def\hr{\hat\r}

\def\ha{\hat a}
\def\hb{\hat b}
\def\hc{\hat c}
\def\hv{\hat v}
\def\hw{\hat w}
\def\rd{\mathrm d}

\def\v{\mathbf{v}}
\def\w{\mathbf{w}}

\begin{document}
\begin{abstract} We show that discrete $W_m$ lattices are bi-Hamiltonian, using  geometric realizations of discretizations of the Adler-Gel'fand-Dikii flows as local evolutions of arc length-parametrized polygons in centro-affine space. We prove the compatibility of two known Hamiltonian structure defined on the space of geometric invariants by lifting them to a pair of pre-symplectic forms on the space of  arc length parametrized  polygons. The simplicity of the expressions of the pre-symplectic forms makes the proof of compatibility straightforward. We also study their kernels and possible integrable systems associated to the pair.

\end{abstract}
\maketitle
\section{Introduction}

This article is motivated by the study of geometric evolution equations for 
continuous and discrete  curves, that are invariant under a given group of symmetries and that induce completely integrable flows on the space of geometric invariants, viewed as coordinates of the moduli space of curves under the action of the symmetry group.  
Many well-known completely integrable partial differential equations, including the sine-Gordon, KdV, MKdV, and NLS equations, have  geometric realizations as curve flows in different geometries (e.g.~hyperbolic for sine-Gordon, centro-affine for KdV, spherical for mKdV, and Euclidean for NLS)  \cite{CI98, P95, DS94, Has72, Lamb80}. 

In the discrete setting, the connection between surfaces and integrable systems has been fairly well explored, as discussed in the monograph by Bobenko and Suris~\cite{BoSu2008},  where discrete surfaces are regarded as two-dimensional layers of multi-dimensional lattices, and the consistency of the underlying equations and geometric properties leads to complete integrable systems and their associated B\"acklund-Darboux transformations.

A general framework for integrable (discrete and continuous) evolutions  of discrete curves, that is not a byproduct of surface theory, is not yet available, although some notable examples have been studied in depth, including the \emph{Pentagram Map}~\cite{OST2013},  
and promising approaches have been proposed, including the use of discrete moving frames to generate Hamiltonian pairs for continuous evolutions of polygons~\cite{MMW, MW}.

The work presented here focuses on a discrete analogue of the so-called $m-1/m$-Adler-Gel'fand-Dikii (AGD) flows, generalizations of the KdV hierarchy first introduced in~\cite{Lax76} and studied in~\cite{Adler7879}. These flows have realizations as local projective flows of curves in $\RP^{m-1}$; within this context, the moving frame approach  leads to a natural description of the associated Hamiltonian structures, defined on the moduli space of projective curves \cite{MB}.
Recently, Mar\'i-Beffa and Wang introduced certain discretisations of the $m-1/m$-AGD flows as well as their realizations as local evolutions of twisted projective polygons in $\RP^{m-1}$~\cite{MW}. By reducing a twisted Poisson structure and a second, \emph{not compatible}, bracket (see \cite{semenov85}) to  the moduli space of twisted  (i.e. quasi-periodic with fixed monodromy) projective polygons, coordinatised in terms of the entries of the Maurer-Cartan matrix of an appropriately chosen moving frame, they also generated two Hamiltonian structures (candidates for discrete $W_m$-algebras). However, a proof of whether these structures were compatible or not remained elusive for general dimension.

In this article, we prove that the pair of Hamiltonian structures defined in~\cite{MB} are compatible in any dimension (Theorem~\ref{compatibility}), and thus the related discrete $m-1/m$-AGD flows are bi-Hamiltonian. We achieve this by lifting the Poisson brackets to a pair $\omega_1, \omega_2$ of pre-symplectic forms defined on the moduli space of arc-length parametrized twisted polygons in centro-affine $\RR^{m}$. The expressions for the pre-symplectic forms are remarkably simple and the general proof of compatibility of the Hamiltonian structures for the $m-1/m$-AGD is straightforward.

At the heart of why shifting the point of view from flows on geometric invariants (curvatures) to geometric flows (polygonal evolutions, invariant under the action of the symmetry group) can help simplify the study of integrability, are several observations that are central to this work. First, Proposition~\ref{leaves} shows that the symplectic leaves of the twisted Poisson bracket are classified by the conjugacy classes of the monodromy of the associated polygon.  Thus, any invariant vector field on the space of polygons will induce a vector field on the geometric invariants that is automatically tangent to the symplectic leaves of the Poisson structure. (Note, however, that the converse is not true: not every flow on periodic geometric invariants will lift to an invariant---monodromy preserving---geometric flow.) Second, as both  $\omega_1$ and $\omega_2$ are trivially reducible to the moduli space, they will induce Poisson structures on the submanifolds of fixed monodromy conjugacy class. Third, Theorem~\ref{kernel} and Corollary~\ref{rigidmotion} show that the kernel of $\omega_1$ is generated by the infinitesimal symmetries and thus $\omega_1$ becomes symplectic when restricted to these submanifolds, suggesting the possibility of generating additional integrable hierarchies beginning with appropriate ``seed" vector fields. We remark, as also pointed out in \cite{MW}, that the {\em Pentagram Map} and its generalizations do not appear to be Hamiltonian with respect to the structures investigated in this work. At this time, connections to other projective polygonal dynamics are yet to be uncovered.

\subsection*{Outline.} In more detail, we summarize the body of the article. In \S2 we illustrate the geometric background, using the 3-dimensional case to introduce notation and key definitions. In \S3 we describe a pair of Poisson brackets on the moduli space of arc length-parametrized twisted polygons, originally introduced in~\cite{MW} in terms of a different, but equivalent, set of coordinates.  \S4 contains the main results of this work. In this section, we define a pair of  pre-symplectic forms that realize the lifts of the Poisson brackets to the space arc length-parametrized twisted polygons (Theorems~\ref{omega2br} and ~\ref{redth}). In Theorem~\ref{compatibility} we present a simple proof of the compatibility of the Poisson brackets as a direct consequence of their representations in terms of the pre-symplectic forms. In \S5 we describe the generators of the kernels of the pre-symplectic forms,  their associated Hamiltonians with respect to their companion form, and we propose a possible approach to the construction of additional integrable hierarchies. The final section contains a discussion of the main results and some open questions.

\section{Background: Evolutions of Polygons in  Centro-affine Space}\label{S2}

It is instructive to consider the $3$-dimensional case first. We define the {\it centro-affine} action of $\SL(3,\RR)$ on $\RR^3$ to be the standard linear action given by multiplication
\( g\cdot x = gx\), where $g\in \SL(3,\RR)$ and $x \in \RR^3$.
Let $\gam: \ZZ \to \RR^3$ be a discrete space curve,  assumed to satisfy
 the following non-degeneracy condition: any three successive vertices $\gam_n$, $\gam_{n+1}$, $\gam_{n+2}$ are linearly independent
as vectors. Although not necessary for parts of the discussion to follow, we also assume that $\gam$ is {\it twisted} (or \emph{quasi-periodic}), that is,
\begin{equation}
\label{mono}
\gam_{n+N} = T \gam_n, \qquad n\in \ZZ,
\end{equation}
for a given matrix $T$ 
in $\SL(3,\RR)$, the \emph{monodromy}, and minimal $N\in \NN$, the \emph{period}. When $T=$Id, the identity matrix,  $\gamma$ is a closed polygon of period $N$.

\subsection*{Moving Frames and Invariants}
Given $\gam = \{ \gam_n \}$ a twisted polygon in $\RR^3$, we introduce the centro-affine moving frame\footnote[2]{See \cite{MMW} for the general definition of discrete moving frame.}  $\rho(\gam) = \big\{\rho_n(\gam)\big\}$ with
\[
\rho_n(\gam) =  \begin{pmatrix}\gamma_n &\gamma_{n+1}& \dfrac{1}{d_n} \gam_{n+2}\end{pmatrix}\in SL(3, \RR),
\]
where the invariant
\[
d_n = |\gam_{n}, \gam_{n+1}, \gam_{n+2} |:=\det( \gam_{n}, \gam_{n+1}, \gam_{n+2} ) 
\]
is the \emph{discrete centro-affine arc length} at the $n$-th vertex. (For sake of convenience, we will use the notation $|A|:=\det(A)$ for matrix determinant.)
From now on we will take $\gam$ to be {\it  arc length parametrized}, i.e.~we assume that $d_n =1$ for all $n$. 

%
The group $G=\SL(3,\RR)$  acts on polygons by diagonal action $g\cdot\gam =  \{g\gam_n\}_{n=1}^N$,  and on the associated moving frame by  $ g \cdot \rho(\gam)= \rho(g\cdot\gam)$, making $\rho$ a  \emph{left} moving frame (i.e. equivariant with respect to the left action of $G$ on itself).

The components of the moving frame $\rho$ satisfy an equation of the form
\begin{equation}\label{rhoK}
\rho_{n+1} = \rho_n K_n,
\end{equation}
where $K_n$ is  the {\it Maurer-Cartan left invariant matrix}. We call~\eqref{rhoK} the Frenet-Serret equations for $\gam$.

 It follows from $|\gam_{n}, \gam_{n+1}, \gam_{n+2}|=|\gam_{n+1}, \gam_{n+2}, \gam_{n+3}|=1$ that $\gamma_{n+3}-\gamma_n$ must be a linear combination of $\gam_{n+1}$ and $\gam_{n+2}$, thus:
\begin{equation}\label{gam3}
\gam_{n+3} = \tau_n \gam_{n+2} + k_n \gam_{n+1}+ \gam_{n},
\end{equation}
where $k_n := |\gam_{n}, \gam_{n+3}, \gam_{n+2}|$ and $\tau_n := | \gam_{n}, \gam_{n+1}, \gam_{n+3}|$ are periodic functions of $n$ of period $N$ and completely determine $\gamma$ up to the action of the group $G$. In other words, $\{ k_n\}$ and $\{ \tau_n\}$ are  a complete set of geometric invariants of the polygon.

From equations~\eqref{rhoK} and~\eqref{gam3} it follows that 
\begin{equation}
\label{K2D}
K_n = \begin{pmatrix} 0 & 0&1\\ 1 & 0 & k_n\\ 0&1&\tau_n \end{pmatrix},
\end{equation}
with $K_{n+N}=K_n$. Moreover, from $\rho_{n+N}=\rho_n K_n K_{n+1} \ldots K_{n+N-1}$,  we find that the product
\[
 K_{1} K_{2} \dots K_N=\rho_1^{-1} T \rho_1
\]
must belong to the conjugacy class of the monodromy. 

\subsection{Evolutions of Polygons and Invariants}
\label{3devol}
Consider a continuous time evolution of the discrete curve, that is also invariant under the centro-affine action (i.e.~group elements take solutions to solutions). Then, the evolution equations take the form:
\begin{equation}\label{vf}
\ddt \gam_n = c_n \gam_n + b_n \gam_{n+1} + a_n\gam_{n+2} = \rho_n \r_n,
\end{equation}
where the components of $\r_n = (c_n, b_n, a_n)^T$,  are functions of the geometric invariants $\{ k_n\}$ and $\{ \tau_n\}$. It follows that the evolution equations of moving frame have the form:
\begin{equation}\label{evolframe3d}
\ddt \rho_n = \rho_n Q_n,
\end{equation}
with
\begin{equation}\label{Q3}
Q_n = \begin{pmatrix} \r_n & K_n \r_{n+1} & K_n K_{n+1} \r_{n+2}\end{pmatrix},
\end{equation}
The expression for $Q_n$ is obtained directly by computing
\[
(\gamma_{n+1})_t = \rho_{n+1} \r_{n+1} = \rho_n K_n \r_{n+1}, \quad\quad (\gamma_{n+2})_t = \rho_{n+2} \r_{n+2} = \rho_n K_nK_{n+1}\r_{n+2},
\]
using equations~\eqref{rhoK} and~\eqref{vf}.

Since $d_n=\det (\rho_n )$,  arc length  is preserved  by the evolution~\eqref{vf} if and only if the trace of  $Q_n$ vanishes for all $n$.  
Introducing the variables
\[
v_n = a_{n+1}, \qquad \quad w_n = b_{n+2} + \tau_{n+1} a_{n+2},
\]
(whose meaning will become clear later on), we compute 
\[
\mathrm{tr}(Q_n) = c_n + c_{n+1} + c_{n+2} + k_n v_n + k_{n+1} v_{n+1}+ \tau_n w_n.
\]
Denoting with $\T$ the left shift operator $\T f_n = f_{n+1}$, the arc length preservation condition $\mathrm{tr}(Q_n)=0$ can be written as:
\begin{equation}\label{c3}
c_n = - \Rc^{-1}\left[(1+\T)k_n v_n + \tau_n w_n\right],
\end{equation}
where the operator 
\[
\Rc:=  1 + \T +\T^2
\]
 is invertible whenever $N\not=3k, k\in \mathbb{N}$ (see Lemma 3.1 in \cite{MW} for a proof of this fact), so we will assume from now on that such is the case.\smallskip

The evolution equations for the geometric invariants is obtained from the compatibility condition  \begin{equation}\label{Kcompat}
\ddt K_n = K_n Q_{n+1} - Q_n K_n
\end{equation}
of the Frenet-Serret equations~\eqref{rhoK} and the frame evolution~\eqref{evolframe3d}.

 Multiplying both sides of~\eqref{Kcompat} by $\small K_n^{-1} = \begin{pmatrix} -k_n&1&0\\-\tau_n&0&1\\ 1&0&0\end{pmatrix}$, we obtain
\[
 \begin{pmatrix} 0&0&(k_n)_t\\ 0&0&(\tau_n)_t\\ 0&0&0\end{pmatrix} =Q_{n+1} -K_n^{-1}Q_nK_n.
 \]
 Setting the $(1,3)$-entries of both sides equal gives
\begin{equation}
\label{invaflow}
\begin{split}
(k_n)_t & =w_{n+1} - (-k_n\ 1 \ 0)\left(\r_n+k_nK_n\r_{n+1}+\tau_nK_nK_{n+1}\r_{n+2}\right)\\
&=w_{n+1} - w_{n-2} + \tau_{n-1} v_{n-1} - \tau_n v_{n+1} + k_n(c_n- c_{n+1}).
\end{split}
\end{equation}

Rewriting $c_n -c_{n+1}= \Rc^{-1}\left[(\T^2-1)k_nv_n + (\T-1) \tau_n w_n\right]$, we get
 \begin{equation}\label{kev3}
(k_n)_t = (\T-\T^{-2})w_n + (\T^{-1}\tau_n-\tau_n\T)v_n+k_n\Rc^{-1}\left[(\T^2-1)k_nv_n + (\T-1) \tau_n w_n\right].
\end{equation}

Similarly, from $(2,3)$-entries of~\eqref{invaflow}, we obtain the evolution of $\tau_n$:
\begin{eqnarray*}
(\tau_n)_t &=& v_{n+2} +k_{n+1}w_{n+1}  - (-\tau_n\ 0\ 1)\left(\r_n+k_nK_n\r_{n+1}+\tau_nK_nK_{n+1}\r_{n+2}\right)\\
&=&v_{n+2}-v_{n-1}+k_{n+1}w_{n+1}-k_nw_{n-1}+\tau_n(c_n-c_{n+2})+\tau_n(k_nv_{n}-k_{n+1}v_{n+1}).
\end{eqnarray*}

\smallskip
Rewriting $c_n-c_{n+2}=-\Rc^{-1}(1-\T^3)k_n v_n-\Rc^{-1}(1-\T^2)(k_nv_n +\tau_nw_n)$ and
\[
\tau_n(k_nv_n-k_{n+1}v_{n+1})=\tau_n\Rc^{-1}\Rc(1-\T) k_n v_n =\tau_n\Rc^{-1}(1-\T)\Rc k_n v_n =\tau_n\Rc^{-1}(1-\T^3)k_n v_n,
\] 
(where we used the fact that $\Rc$ and $\T$ commute), we obtain:
 \begin{equation}\label{tauev3}
 (\tau_n)_t = (\T^2-\T^{-1})v_n + (\T k_n-k_n\T^{-1})w_n+\tau_n\Rc^{-1}(\T^2-1)(k_nv_n +\tau_nw_n).
 \end{equation}
 The evolution equations for the geometric invariants can be written as:
 \begin{equation}\label{ktev3}
 \begin{pmatrix} k_n\\ \tau_n\end{pmatrix}_t = \Pc_{1n} \begin{pmatrix} v_n \\ w_n\end{pmatrix}.
 \end{equation}
where
 \begin{equation}\label{Poisson3}
 \Pc_{1n} = \begin{pmatrix} \T^{-1}\tau_n-\tau_n\T +k_n\Rc^{-1}(\T^2-1)k_n& \T-\T^{-2}+k_n\Rc^{-1}(\T-1) \tau_n \\\\  \T^2-\T^{-1}+\tau_n\Rc^{-1}(\T-1)k_n&\T k_n-k_n\T^{-1}+\tau_n\Rc^{-1}(\T^2-1)\tau_n
 \end{pmatrix}
 \end{equation}
 A  gauge equivalent version of  $\Pc_{1n}$, as well as its higher-dimensional generalizations, were shown to be Poisson operators in~\cite{MW}.  In the next section, we describe how such discrete Hamiltonian structures arise in general dimension, together with its bi-Hamiltonian companions. 
 

\section{Hamiltonian structures on the moduli space of polygons}\label{Poisson}

Relying on the construction in \cite{MW} for polygons in projective space, this section introduces two Hamiltonian structures on the moduli space of arc length-parametrized twisted polygons  in centro-affine $\RR^m$ that are adapted to our particular choice of moving frame.  We first discuss the equivalence between the centro-affine and projective cases.
\begin{definition}
A projective polygon $\{p_n\}_{n=1}^N$, $p_n\in \RP^{m-1}$, is \emph{non-degenerate} if there exists a lift  $\{\gam_n\}$, $\gam_n\in \RR^m$, of $\{p_n\}$ such that $|\gam_n,\gam_{n+1},\dots,\gam_{n+m-1}|\ne 0, \forall n$. 
\end{definition}
This property does not depend on the particular choice of lift: in fact, replacing $\gam_n$ with $\lambda_n \gam_n$,  $\lambda_n\ne0$, changes the determinant by an overall non-zero factor. The quantity
\begin{equation}\label{dn}d_n = |\gam_n, \gam_{n+1}, \dots, \gam_{n+m-1}|
\end{equation}
is the \emph{centro-affine arc length} of $\{\gamma_n\}$ at its $n$-th vertex. Let $\M_N$ be the set of non-degenerate twisted polygons in $\RR^m$, of given period $N$ and monodromy $T$, and let
\[
\mathcal{M}_N^1: = \left\{ \gamma:  \ZZ \rightarrow \RR^m \, | \ \gam\in \M_N,\ d_n=1, \forall n\right\}
\]
be the subset of arc-length parametrized polygons.
\begin{definition}
The moduli space of $\M^1_N$ is the set of all orbits of $\M^1_N$ under the action of the projective group $PSL(m,\RR)$, that is $\M^1_N/PSL(m,\RR)$.
\end{definition}
\begin{proposition}\label{diffeo}
If $N$ and $m$ are co-prime, the moduli space  of $\M_N^1$  is diffeomorphic to the moduli space of non-degenerate twisted projective polygons in $\RP^{m-1}$ with the same monodromy.
\end{proposition}
\begin{proof} We use the known fact that a non-degenerate projective polygon in $\RP^{m-1}$ admits a unique lift  $\gam = \{\gam_n\}$ to $\RR^m$ satisfying
$|\gam_{n}, \gam_{n+1}, \dots, \gam_{n+m-1}| = 1\, \forall n$,  provided $N$ and $m$ are co-prime. (For a proof, see, e.g.~\cite{MW}.)


Clearly,  $\gam$  is an arc length parametrized twisted polygon in centro-affine $\RR^m$, since the action of $\mathrm{PSL}(m,\RR)$ on projective polygons becomes the linear action of $\SL(m,\RR)$ on lifts. The  invariants  $\{k_n\}$ and $\{ \tau_n\}$ defined in terms of the Frenet-Serret equations are both the projective invariants of $\{p_n\}$ as discussed in \cite{MW}, as well as the centro-affine invariants of its lift $\{\gam_n\}$ to $\RR^m$. \end{proof} 

The inverse of the correspondence described in Proposition~\ref{diffeo} is simply the projectivization of a centro-affine arc length parametrized $N$-gon. It follows that the invariants can be used as coordinates for either moduli space, regarding them as functions of $\{p_n\}$ in the projective case, and of $\{\gam_n\}$ in the centro-affine case. \smallskip

Next, we use Proposition~\ref{diffeo} and the results  for projective $N$-gons in~\cite{MW} to define two Hamiltonian structures on the moduli space of centro-affine arc length parametrized $N$-gons. These structures arise as reductions of more general structures as described below in the context of Poisson-Lie groups.

\subsection{A twisted Poisson structure on a Poisson-Lie group}
Let $G$ be a semisimple Lie group and $\g$ be its Lie algebra. Assume $\g$ has a non-degenerate inner product $< \ , \ >_\g : \g \times \g^* \rightarrow \RR$ that identifies $\g$ with its dual $\g^\ast$ (for matrix Lie algebras, this is usually the trace of the product of two Lie algebra elements.)
\begin{definition}[Left and right gradients] \label{gradients}
Let $G^N$ denote the Cartesian product of $N$ copies of $G$. The group $G^N$ acts on itself by diagonal (i.e.~component-wise) action: $g\tilde{g}:=(g_{n}\tilde{g}_{n})$, where $g, \tilde{g} \in G^N$.

Let $\F: G^{N} \to \RR$ be a differentiable function. The \emph{left gradient} of $\F$ at $L = \{L_n\}\in G^{N}$  is the element $\nabla \F(L) = \{\nabla_n \F(L)\}$ of $\g^{N}$ defined by
\[
\left. \frac d{d\epsilon}\right|_{\epsilon = 0} \F(\exp(\epsilon \xi) L) : = \langle \nabla \F(L), \xi\rangle, \qquad \forall \xi\in \g^N,
\]
where
$\displaystyle \langle \xi , \tilde{\xi} \rangle=\sum_{n=1}^N <\xi_n, \tilde{\xi}_n>_\g$ is the induced inner product on $\g^N$. 

Similary, the {\it right gradient} of $\F$ at $L$, $\nabla'\F(L)\in g^N$, is defined by
\[
\left. \frac d{d\epsilon}\right|_{\epsilon = 0} \F(L\exp(\epsilon \xi) ): = \langle \nabla' \F(L), \xi\rangle, \qquad \forall \xi\in \g^N.
\]
(See \cite{FRS, MW}.)
 \end{definition}
Since $L_n\exp(\epsilon \xi_n)=L_n\exp(\epsilon \xi_n)L_n^{-1}L_n=\exp(\epsilon L_n \xi_n L_n^{-1})L_n$, one gets $\langle \nabla' \F(L), \xi\rangle = \langle\nabla \F(L), L\xi L^{-1}\rangle$ and the following relation between the right and left gradients:
 \begin{equation}\label{lrgrad}
 \nabla'\F(L) = L^{-1} \nabla \F(L) L,
 \end{equation}
 expressed component-wise as $ \nabla_n'\F(L) = L_n^{-1} \nabla_n \F(L) L_n$, where the symbol $\nabla_n$ (resp. $\nabla_n'$) denotes the $n$-th component of the left (resp. right) gradient.

The author of \cite{semenov85} defined what is known as a {\it twisted Poisson bracket}, a definition that was reinterpreted in \cite{FRS} using their particular notation. In this paper we will adopt the definition in \cite{FRS}.   
\\
Assume that  $\g$ has a grading $\g = \g_+\oplus\h_0\oplus\g_-$, with $\h_0$ commutative and $\g_+$ dual to $\g_-$. Let $\xi=\xi_++\xi_0+\xi_-$ be the associated decomposition of $\xi\in \g$. Let $r\in \g\otimes\g$ represent an $r$-matrix, as described in  \cite{FRS}, page 616. In particular, $r$  satisfies the classical Yang-Baxter equation
\begin{equation}\label{YM}
[r_{12}, r_{13}]+[r_{12},r_{23}]+[r_{13},r_{23}] = 0,
\end{equation}
where $r_{12} = \phi(r)$ with $\phi(a\otimes b) = a\otimes b\otimes 1$ is linear, with $r_{13}$ and $r_{23}$ similarly defined. \smallskip

Given $\F, \H$ smooth scalar-valued functions on $G^N$ and $L\in G^N,$ the \emph{twisted Poisson bracket} is given by \cite{FRS}:
\begin{equation}\label{twisted}
\begin{split}
\{\F, \H\}(L) & := \sum_{s=1}^Nr(\nabla_s \F \wedge \nabla_s \H) + \sum_{s=1}^N r(\nabla_s' \F \wedge \nabla_s' \H)  \\ &- \sum_{s=1}^N r\left( (\T\otimes 1)(\nabla_s' \F \otimes \nabla_s \H)\right) + \sum_{s=1}^N r  \left((\T\otimes 1) (\nabla_s' \H \otimes \nabla_s \F)\right), 
\end{split}
\end{equation}
where $\xi\wedge \eta=\tfrac12( \xi\otimes\eta - \eta\otimes\xi)$.
Equation~\eqref{twisted} defines a Hamiltonian structure on $G^{(N)}$, as shown by Semenov-Tian-Shansky~ \cite{semenov85}.  Moreover, the {\em twisted} gauge action of $G^{N}$ on itself---mapping $\{L_n\}$ to $\{g_{n+1}L_n g_n^{-1}\}$---is a Poisson map and its orbits coincide with the symplectic leaves \cite{FRS, semenov85}.
 These brackets play a central role in the following construction.
 
\subsection{Poisson structures on the moduli space of centro-affine polygons}\label{PSonM}

The symmetry group of the space of arc length parametrized twisted  polygons in centro-affine $\RR^m$ is  $\, G=\SL(m, \RR)$. In this case, the $r$-matrix and the ensuing twisted Poisson bracket~\eqref{twisted} are defined in terms of the standard gradation $\g = \g_+\oplus\h_0\oplus\g_-$ of its Lie algebra  $\sl(m, \RR)$. Here, $\g_-, \g_+$ and $\h_0$ are the subsets of (respectively) strictly lower triangular, strictly upper triangular, and traceless diagonal matrices. The standard $r$-matrix for $\sl(m)$ is defined by
\[
r = \sum_{i<j} E_{i,j}\otimes E_{j,i} + \sum_{k=1}^{m-1} F_k\otimes F_k,
\]
where $E_{i,j}$ is the matrix with $1$ in entry $(i,j)$ and zero elsewhere, and where the $F_k$'s are appropriately chosen generators of $\h_0$, to ensure that $r$ satisfies \eqref{YM}. The exact form of the $F_k$'s is not relevant, since they do not play a role in our construction. Note that 
\begin{equation}\label{rmatrix}
r(\xi\otimes \eta) = \langle \xi_-, \eta_+\rangle + \sum_k\langle \xi_0^k, \eta_0^k\rangle, 
\end{equation}
where $\xi_0^k$ are the diagonal components of $\xi$ dual to $F_k$. \smallskip

The Lie algebra $\sl(m, \RR)$ admits a second gradation $\g = \g_1\oplus \g_0\oplus\g_{-1}$, with
\begin{equation}\label{gradation}
\g_{1}= \begin{pmatrix} 0&\ast&\dots&\ast\\0&0&\dots&0\\ \vdots&\vdots&\vdots&\vdots\\ 0&0&\dots&0\end{pmatrix},\quad \g_0 = \begin{pmatrix}  \ast&0&\dots&0\\0&\ast&\dots&\ast\\ \vdots&\vdots&\vdots&\vdots\\ 0&\ast&\dots&\ast\end{pmatrix},\quad \g_{-1} = \begin{pmatrix} 0&0&\dots&0\\ \ast&0&\dots&0\\ \vdots&\vdots&\vdots&\vdots\\ \ast&0&\dots&0\end{pmatrix}.
\end{equation}
This second gradation realizes  the quotient $G/H$ as a homogeneous space, where $H\in \SL(m, \RR)$ is the subgroup with  Lie algebra $\h = \g_1\oplus\g_0$, acting on $G$ by left multiplication.
\smallskip

\noindent{\bf Remark.} The two gradations of $\sl(m, \RR)$ are \emph{compatible}, in the sense that $\g_1\subset \g_+$ and $\g_{-1}\subset \g_-$. 
\smallskip

We will also consider  the quotient space $G^N/H^N$, where $H^N$ acts on $G^N$ via the \emph{right discrete gauge action}: 
\begin{equation}\label{dgauge}
(h_n , g_n) \to h_{n+1} g_n h_{n}^{-1}.
\end{equation}

Following a similar argument as the one given in Section 5 of~\cite{MW}, and using Proposition~\ref{diffeo}, we identify (locally) the moduli space of $\M^1_N$  with $G^N/H^N$ and coordinatize it by means of the geometric invariants of the polygons. The argument has two main parts.\smallskip

\noindent
I. \emph{Group description of twisted polygons.} Consider the space $(G/H)^N$ and denote its elements with $\{[g_n]\}$: $N$-tuples of equivalence classes with respect to the left action on $H$. First note that, if $\e_1 = (1 \ 0 \ \dots \ 0)^T$ and $h\in H$, then  $h \e_1 = a \e_1$ for some scalar $a\not=0$. This is easily verified by using the general form of $h \in H$. Letting  $\hat g_n=g_nh_n$ be another element of $[g_n]$,  then $\hat g_n \e_1=g_n h_n \e_1=ag_n\e_1$. Therefore,  $[\hat g_n] = [g_n]$ if and only if the projective class of $\{ g_n \e_1 \}$ is the same as that of $\{\hat g_n\e_1\}$.

Given $\{[g_n]\}\in (G/H)^N$ and $T\in G$,  we extend $g_n$ quasi-periodically as $g_{n+N}=Tg_n$, and let $\gam=\{ \gamma_n \}$ be the unique lift of the projective class of $\{ g_n \e_1\}$ to $\RR^m$ satisfying $|\gam_n,\dots,\gam_{n+m-1}| = 1$  and $\gam_{n+N}=T\gam_n$. (See also Proposition~\ref{diffeo}.) In this way we construct a unique $\gam\in \M^1_N$ of monodromy $T$ from any given element of $(G/H)^N$, as illustrated in the following diagram:
\[
\begin{array}{ccccc}
(G/H)^N & \longrightarrow  & (\RP^{m-1} )^N& \longrightarrow& \M_N^1 \\
& & & \\
 {[\{g_n\}]} & \longrightarrow   & \{[g_n\e_1]^\pi\} & \longrightarrow &\{ \gam_n \},
\end{array}
\]
where $[g_n\e_1]^\pi$ denotes the projective class of the first column of $g_n$.

Note that, if  $\rho=\{\rho_n\}$, $\rho_n = (\gamma_n,\dots, \gam_{n+m-1})$, is the  {left-moving frame} of $\gam$, then  $\{[\rho_n]\}= \{[g_n]\}$, since the first columns of $\rho_n$ and $g_n$ belong to the same projective class.  In this way, an arc length parametrized twisted polygon $\gam$ is identified with the equivalence class of its left-moving frame $\rho$, that is, $\M_N^1$ is (at least) locally\footnote[2]{We will not need more than local identification.}
identified with $(G/H)^N\times G$.
We remark that a section of the quotient $(G/H)^N$ in a neighborhood of the identity can be locally identified with $G_{-1}$, as illustrated in the following diagram:
\[
\begin{array}{cccc}
(G/H)^N & \longrightarrow (\RP^{m-1})^N & \longrightarrow& G_{-1} \\
& & & \\
{ [\{g_n\}]} & \longrightarrow  {\{[g_n\e_1]^\pi\}} & \longrightarrow & {\begin{pmatrix} 1 & 0^T\\ \v & \I_{m-1}\end{pmatrix}}
\end{array}
 \qquad \quad \begin{pmatrix}1\\\v \end{pmatrix}\in [g_n\e_1]^\pi.
\]

\noindent
\emph{Group description of the moduli space}. Next, consider the quotient space $G^N/H^N$, where $H^N$ acts on $G^N$ via the right discrete gauge action~\eqref{dgauge}, and introduce the map
\begin{equation}\label{Kmap}
\begin{array}{ccc}
\K: (G/H)^N \times G & \longrightarrow & G^N/H^N\\  
& & \\
\big(\{[g_n]\}, T\big) &  \longrightarrow &\big[\{g_{n+1}^{-1}g_{n}\}\big],
\end{array}
\end{equation}
where we define $g_{n+N} := T g_n$.\footnote[3]{Right invariance is used in the description of  $G^N/H^N$, since the twisted bracket on this quotient space is defined for right gauges; while left invariance is used for $(G/H)^N$, as it is standard for homogeneous spaces.} If $m_n= g^{-1}_{n}g_{n+1}$, then $\{m_n\}$  is periodic of period $N$ and $m_1m_2 \ldots m_N=T$.    The map $\K$ is well-defined since, if $[\hat g_n]=[g_n]$, i.e. $\hat g_n = g_n h_n$, and  both are extended with the same monodromy $T$, then $ \mathcal{K}(\{[\hat{g}_n]\}, T)=[\{h_{n+1}^{-1}g_{n+1}^{-1}g_{n}h_{n}\}]=\K([{g}_n], T)$.  Also,  if $[\{m_n^{-1}\}] \in G^N/H^N$ has a representative $m$ satisfying $m_1m_2 \ldots m_N=T$ and $\hat{m}^{-1}\in [\{m_n^{-1}\}]$ is a different representative of the same equivalence class, then $\hat{m}_1\hat{m}_2\ldots \hat{m}_N=h^{-1}_1T h_{N+1}=h^{-1}_1T h_{1}$ belongs to the $H$-conjugacy class $\mathcal{C}_T$ of $T$.   It follows that $\K$ is a surjective map; in fact, given $[\{m_n^{-1}\}]\in G^N/H^N$, and choosing an arbitrary representative $\{m_n^{-1}\}$ with $m_1m_2 \ldots m_N=T$, we can solve recursively $g_{n+1} = g_n m_n$, and recover an $N$-gon of monodromy $T$, up to an overall left action by the symmetry group, amounting to the selection of $g_1$. 
\begin{remark}
It is a general fact that the cardinality  of the conjugacy class of $T$ is the index of the isotropy subgroup of $T$. It follows that, there is a natural foliation of $G^N/H^N$, where the dimension of the stabilizer of the monodromy of an element in $G^N/H^N$ equals the co-dimension of the leaf of that element. This foliation coincides locally with the symplectic foliation of the reduced bracket defined below, described in Proposition~\ref{leaves}. 
\end{remark}


Finally, let  $\gam \in \M_N^1$ define a local section of $(G/H)^N$ by means of the equivalence class representative given by its left-moving frame $\rho$. Let $\{K_n\}:=\{\rho_n^{-1}\rho_{n+1}\}$ be its left Maurer-Cartan matrix. Since $\mathcal{K}\big(\{[\rho_n]\}, T\big) = \big[\{\rho_{n+1}^{-1}\rho_n\}\big] =  [\{K_n^{-1}\}]:=[\{L_n\}]$,  $\{L_n\}$ defines a smooth local section of $G^N/H^N$. On the other hand, the Maurer-Cartan matrix of a twisted $N$-gon in $\RR^m$ is uniquely determined by the geometric invariants of the polygon (see (\ref{K})). Thus, $G^N/H^N$ can be locally identified  with the moduli space of $\M_N^1$ (or, by projectivizing, with the moduli space of twisted polygons in $\RP^{m-1}$) and the geometric invariants can be used as invariant coordinates on the moduli space. We will return to the map $\K$ later on.

When gradations are compatible,  the Poisson bracket (\ref{twisted}) - defined via the first gradation - can be reduced to the quotient $G^{N}/H^{N}$,  to obtain
 a Poisson structure on the moduli space. The following paraphrases Theorem 5.5 in~\cite{MW}.
\begin{theorem}\label{bracketth} Given $G=\SL(m,\RR)$, let $\g = \g_+\oplus\h_0\oplus\g_-=\g_1\oplus \g_0\oplus\g_{-1}$ be the two compatible gradations of its Lie algebra described above. Let $H$ be the subgroup with Lie algebra $\h = \g_0\oplus\g_1$, and assume  $H^N$ acts on $G^N$ via the gauge action~\eqref{dgauge}. We also require $g_{n+N} = g_n, \, \forall n$. Then, 
the twisted Poisson structure \eqref{twisted} defined on $G^{N}$ with $r$ as in \eqref{rmatrix} is locally reducible to the quotient $M = G^N/H^N$.
\end{theorem}
We will denote the reduce bracket with $\{ \ , \ \}_R$.
\begin{remark} 
We note that Semenov's twisted Poisson bracket\eqref{twisted} is defined in terms of the right action instead of the left action. For computations explicitly involving \eqref{twisted} (such as in Section~\ref{polygev}), we will then need to work with smooth sections of $M$ given by right Maurer-Cartan matrices (the inverses of the left Maurer-Cartan associated to our choice of frame), in order to preserve each symplectic leaf and be able to reduce the bracket to the quotient space. This is the price we pay for using a more ``canonical'' choice of moving frame. 
\end{remark}

A second Poisson bracket arises from a companion bracket to \eqref{twisted}  found in \cite{MW}, as a reduction of the tensor:
\begin{equation}
\label{theta}
\theta(\F, \H)(L) = \frac 12\sum_{n=1}^N \langle(\nabla'_n \F)_{-1}, (\nabla'_n\H)_{1}\rangle -\langle(\nabla'_n \H)_{-1}, (\nabla'_n\F)_{1}\rangle ,
\end{equation}
to $M=G^N/H^N$. Here, $\F, \H : G^N\rightarrow \RR$ are differentiable functions. While tensor $\theta$ is not a Poisson bracket on $G^N$, its reduction to $M$ is a Poisson bracket in any dimension $m$, as shown in \cite{MW}. The origin of this second bracket, which will be denoted by $\{\ , \ \}_0$, has remained elusive and proof of its compatibility with the reduced twisted bracket could only be achieved in dimension $m=3$. 

As consequences of our main results from Section 4, we will show that $\theta$ can be interpreted as the invariant version of a Poisson bracket at the curve level, explaining why its reduction to the moduli space is a Poisson bracket. We will also provide a straightforward proof of the compatibility of the two Poisson structures in any dimension (Theorem~\ref{compatibility}).

{\begin{proposition}\label{leaves}
The symplectic leaves of the reduced bracket  from Theorem \ref{bracketth} are locally classified by the conjugacy class of the associated monodromy.
\end{proposition}
\begin{proof}
The symplectic leaves of \eqref{twisted} are gauge orbits of the discrete action of $G^N$ on itself~\cite{FRS}; therefore,  the symplectic leaves of the reduced bracket $\{ \ , \ \}_R$ are the projection of the gauge orbits. This is only true locally, since the reduction itself is local, and  additional {\it global} discrete invariants may exist for the projected leaves.\\ 
Let $\{m_n\}$ and $\{\hat{m}_n\}$ belong to the same gauge orbit, i.e. $\hat{m}_n=g_n^{-1} m_n g_{n+1}$, $g_n\in G$ with $g_{n+N}=g_n$, and let $T=m_1m_2 \ldots m_N$ be the monodromy of $\{m_n\}$. Then $\hat{m}_1\hat{m}_2\ldots \hat{m}_N=g^{-1}_1T g_{N+1}=g^{-1}_1T g_{1}$ belongs to the conjugacy class of $T$.  


Conversely, given $\{m_n\}$ with $m_1m_2 \ldots m_N=T$, for any $T^g=gTg^{-1}$ in the conjugacy class of $T$ we can find $\{[g_n]\}\in (G/H)^N$ such that $\mathcal{K}(\{[g_n]\}, T^g) = [\{m_n\}]$. In fact, $g_n$ can be constructed recursively from $g_{n+1}=g_n m_n$ with $g_1=g$. Then, $g_{N+1}=g_1 T=(gTg^{-1})g_1$, i.e.,  locally, elements of each leaf---{\it before projection onto $G^N/H^N$}---have the same monodromy class; the same will remain locally true on the quotient $G^N/H^N$ after projection.
~
\end{proof}}
The relationship between Hamiltonian flows induced by the reduced bracket $\{ \ , \ \}_R$ and polygon evolutions in centro-affine $\RR^m$ is described next.
 
 \subsection{Polygon evolutions and Hamiltonian flows in the moduli space}\label{polygev}
 
We first describe the set-up for arbitrary dimension $m$; much of this is adapted from~\cite{MW} for a different choice of moving frame.

Let $G$  be a Lie group acting on a manifold $M$ by left (or right) action.  A vector field $V$ on $M$ is \emph{left} (or \emph{right}) \emph{invariant} if the associated flow $\Phi_t$ is  equivariant with respect to the left (or right) group action, i.e. $g\cdot \Phi_t(x) = \Phi_t(g\cdot x), \, g\in G, x\in M.$ Let $\V_N\subset TM_N$  denote the \emph{left invariant vector fields} on $\M_N$, and let $\V_N^1\subset T\M_N^1$ be the subspace of arc length-preserving left invariant vector fields. (Recall that $\M_N^1\subset \M_N$  is the subset of non-degenerate $N$-gons parametrized by centro-affine arc length.)

\subsubsection{Evolution of moving frames.} Given $\gam= \{ \gam_n \}\in \M_N^1$, we introduce the left moving frame $\rho:\M_N^1 \to \SL(m,\RR)^N$, $\rho(\gam) = \{\rho_n(\gam)\}_{n=1}^N$ with components
\begin{equation}\label{leftmovframe}
\rho_n(\gam) = (\gam_{n}, \gam_{n+1}, \dots, \gam_{n+m-1}).
\end{equation}
As discussed in Section~\ref{S2}, this moving frame is gauge equivalent to the one used in~\cite{MW}, thus many of the results  in that work will apply to our context. In particular, the generating set of geometric invariants will be used to coordinatize the moduli space. The geometric invariants are defined via the following relation:
\begin{equation}\label{as}
\gam_{n+m} = a_n^{m-1} \gam_{n+m-1} + \dots+ a_n^1\gam_{n+1}+(-1)^{m-1} \gam_n,
\end{equation}
where the coefficient of $\gamma_n$ comes from specializing $(-1)^{m-1}d_{n+1}/d_n$, with $d_n$ as in (\ref{dn}), to  $\gamma \in \M_N^1$. This gives
\begin{equation}
\label{invariantsa}
a^r_n=|\gam_n,\dots,\gamma_{n+r-1}, \stackrel{r+1}{\overbrace{\gamma_{n+m}}}, \gamma_{n+r+1},\dots,\gamma_{n+m-1}|.
\end{equation}
The discrete equivalent of the Serret-Frenet equations is:
\begin{equation}\label{FreSer}
\rho_{n+1} = \rho_n K_n, 
\end{equation}
\begin{equation}
\label{K}
K_n = \begin{pmatrix} \0^T& (-1)^{m-1}\\ \I_{m-1}&\a_n\end{pmatrix};
\end{equation}
where $\I_{m-1}$ is the $(m-1)\times (m-1)$ identity matrix, $\0$ the zero vector in $\RR^{m-1}$, and $\a_n=\begin{pmatrix}a_n^1 \ a_n^2 & \ldots & a_n^{m-1}\end{pmatrix}^T$  the vector of invariants at vertex $\gamma_n$. We will also need the inverse of the Maurer-Cartan matrix $K$: 
\begin{equation}\label{L}
K^{-1}:=L=\{L_n\}, \qquad \qquad L_n = \begin{pmatrix} (-1)^m \a_n & \I_{m-1}\\ (-1)^{m-1}& \0^T\end{pmatrix}.
\end{equation}

The components of an invariant vector field $X = \{ X_n\}_{n=1}^N \in \V_N^1$ can be expressed as linear combinations of the components of $\gamma$ whose coefficients only depend on the geometric invariants:
\begin{equation}\label{X}
X_n = \sum_{\ell=0}^{m-1} r_n^\ell \gam_{n+\ell} = \rho_n \r_n, \qquad\quad \ \r^T_n = (r_n^0, \dots , r_n^{m-1}).
 \end{equation}
The induced evolution on $\gamma$ is (component-wise):
 \begin{equation}\label{pev}
 \frac{\rd}{\rd t}\gam_n = X_n.
 \end{equation}
From~\eqref{X} and~\eqref{FreSer}, one gets $\displaystyle \ddt\gam_{n+1} = \rho_{n+1}\r_{n+1} = \rho_n K_n\r_{n+1}$ and, in general, $\displaystyle \ddt\gam_{n+\ell}= \rho_nK_nK_{n+1}\dots K_{n+\ell-1}\r_{n+\ell}$. Thus,  the evolution of the left-moving frame is given by
 \begin{equation}\label{frameev}
\frac{\rd}{\rd t}\rho_n=  \rho_nQ_n, 
 \end{equation}
where 
\begin{equation}\label{Q}
 Q_n = \begin{pmatrix}\r_n & K_n\r_{n+1}& K_n K_{n+1} \r_{n+2}& \dots& K_nK_{n+1}\dots K_{n+m-2}\r_{n+m-1}\end{pmatrix}.
 \end{equation}
The Frenet-Serret equations~\eqref{FreSer} and the frame evolution~\eqref{frameev} are compatible provided
 \begin{equation}\label{compat}
 \ddt K_n = K_n Q_{n+1} - Q_n K_n.
 \end{equation}
 Arc length preservation during the evolution amounts to $|\gam_n \ \gam_1 \ \ldots \gam_{n+m-1}|=1$ for all times. This condition is more conveniently expressed as 
 \begin{equation}
 \label{alpreserve}
\mbox{tr}(Q_n)=0,
 \end{equation} 
 and uniquely determines  $r_n^0$ in terms the other entries of $\r$ as shown in the following:

 \begin{proposition}\label{LAP}
For $N$, $m$ co-prime integers, and  $X$ as in~\eqref{X}, $X\in \V_N^1$ if and only if $r_n^0  = R_n(\r)$, where  $R_n$ is an algebraic function of the $r_s^i$'s.
\end{proposition}
\begin{proof} 
From~\eqref{Q} we compute
\begin{equation}\label{traceQ}
\text{tr}(Q_n)=\sum_{s=1}^{m-1} \e_{s+1}^TK_n\dots K_{n+s-1}\r_{n+s} + r_n^0,
\end{equation}
where the coefficient of $r_{n+s}^0$ is obtained as follows:
\[
\e_{s+1}^TK_nK_{n+1}\dots K_{n+s-1}\e_1 = \e_{s+1}^TK_n\dots K_{n+s-2}\e_2= \ldots=\e^T_{s+1} K_n \e_s = \e_{s+1}^T \e_{s+1} = 1.
\]
Then $X\in \V_N^1$ if and only if 
\[
\text{tr}(Q_n)=(1+\T+\dots+\T^{m-1}) r_n^0 + \widehat{R}_n=0,
\]
where $\widehat{R}$ is a function of  $\r^i_{n+s}$, $i, s=1, \ldots m-1$. The operator $\Rc=1+\T+\dots+\T^{m-1}$ is invertible if and only if $N$ and $m$ are co-prime (see Lemma 3.1 in~\cite{MW}). The claim follows, since its inverse, though non-local, can be expressed as an algebraic function of $\T$.
\end{proof}

\subsubsection{Poisson brackets} In the rest of this section,  $G=\SL(m, \RR)$ and we will work with two gradations of its Lie algebra: the standard gradation $\g = \g_+\oplus\h_0\oplus\g_-$ and the gradation $\g=\g_1\oplus \g_0 \oplus \g_{-1}$, 
both described in Section~\ref{PSonM}. 

For $H$ the subgroup whose Lie algebra is $\h=\g_1\oplus \g_0$, let $f:G^N/H^N\to \RR$ be a smooth functional on the moduli space and $\F: G^N \to \RR$ an extension of $f$ to $G^N$, invariant under the discrete gauge action~\eqref{dgauge} of $H^N$.

 As discussed in Section~\ref{PSonM}, we will work with the smooth sections of  the moduli space  $M=G^N/H^N$ defined by the right Maurer-Cartan matrices $L=\{L_n\}$ as in \eqref{L}, using the invariants $a_n^j$, $j=1, \ldots, m-1, n=1, \ldots , N,$ as coordinates of  $M$.
Here $\langle \ , \ \rangle$ denotes the standard inner product on $\g^N$ defined by
\begin{equation}
\label{inprodN}
\langle \xi,\tilde{\xi} \rangle:=\sum_{n=1}^N \mbox{tr}(\xi_n\tilde{\xi}_n), \qquad
\xi, \tilde{\xi} \in \g^N.
\end{equation}

\begin{lemma} For fixed $n$, the vector $\delta_nf$ of variational derivatives of $f$ with respect to the geometric coordinates $a_n^j$'s, $j=1, \ldots, m-1$, and the left gradient of its extension $\F$ satisfy the following relation:
\begin{equation}\label{onecond}
\nabla_n\F = \begin{pmatrix} \ast&\ast  \\-\delta_nf^T  & \ast\end{pmatrix}.
\end{equation}
Moreover, the $\g_1$-component of the right gradient satisfies:
\begin{equation}\label{onecomponent}
(\nabla'_n\F)_{1} = \begin{pmatrix} 0& (-1)^m\delta^T_n f\\ \0 & \0_{m-1}\end{pmatrix},
\end{equation}
\end{lemma}
\begin{proof}
Recall that the \emph{left gradient} of $\F$ at $L = \{L_n\}\in G^{N}$ is defined by: 
\[
\left. \frac d{d\epsilon}\right|_{\epsilon = 0} \F\big(\exp(\epsilon \xi) L\big) : = \langle \nabla \F(L), \xi\rangle, \qquad \forall \xi \in \g^N.
\]
In particular, choosing $\xi$ of the form $V^{(k)}:=\{V^{(k)}_n\}$ with $V^{(k)}_n=\begin{pmatrix}\ \0_{m-1} & -\v_n \\ \0^T&  0 \end{pmatrix}\delta_{nk}$,  we compute
\[
\exp(\epsilon V^{(k)}_n) L_n=\begin{pmatrix}\I_{m-1}&-\epsilon\v_n\delta_{nk}\\ \0^T&1\end{pmatrix}L_n = \begin{pmatrix} (-1)^m(\a_n+\epsilon\v_n\delta_{nk})&\I_{m-1}\\ (-1)^{m-1}&\0^T\end{pmatrix}.
\]
Since $\F$ is an extension of a function defined on the moduli space, we must have that $\F\big(\exp(\epsilon V^{(k)}) L\big) = f\big(\{\a_n+\epsilon\v_n\delta_{nk}\}\big)$. Differentiating both sides with respect to $\epsilon$ and evaluating at $\epsilon=0$, we obtain
\[
\delta_k f\cdot\v_k=\langle \nabla \F, V^{(k)}\rangle=\sum_{n=1}^N \mbox{tr}(\nabla_n \F V^{(k)}_n)=\mbox{tr}(\nabla_k \F V^{(k)}_k),
\]
where $\delta_k f\cdot\v_k$ is the dot product of $\v_k$ with the vector of variational derivatives of $f$ with respect to the components of $\a_k$. Writing $\nabla_k \F$ in block form $\begin{pmatrix} A_{m-1} &\mathbf{b} \\ \mathbf{c}^T & d \end{pmatrix}$,
we compute the right-hand side as $\langle \nabla \F, V^{(k)}\rangle=-\mathbf{c}^T\cdot \v_k$ and (after replacing $k$ with $n$) relation~\eqref{onecond} follows. Finally, \eqref{onecomponent} can be obtained by direct verification using the relation \eqref{lrgrad} between the left and right gradient.
\end{proof}
%
%
%

\begin{lemma}
Let $\F$ be an invariant extension of a smooth functional $f:G^N/H^N\to \RR$. Let $\T$ denote the left shift operator. Then
\begin{equation}\label{maincond}
\T^{-1}\nabla_n \F - \nabla_n'\F \in  \g_1,
\end{equation}
\end{lemma}
\begin{proof}

The gauge action of $H^N$ on $K$:  $K_n \to h_n^{-1} K_n h_{n+1}$  component-wise, induces the action $L_n \to h_{n+1}^{-1} L_n  h_n$ on the inverse $L=K^{-1}$.
Since $\F$ is an invariant extension of  $f$, it must be constant on the gauge orbits,
 i.e. $\F\big(h\cdot L\big) =\F\big(\{h_{n+1}^{-1} L_n h_{n}\}\big) = \F(L), \forall h \in H^N$.  
Choosing $h=\exp(\epsilon \xi)$ with $\xi_{n+N}=\xi_n$,
 \[
 \begin{split}
0 & =\left. \frac d{d\epsilon}\right|_{\epsilon = 0} \F\big(\{\exp(-\epsilon \xi_{n+1}) L_n\exp(\epsilon \xi_n)\}\big) \\ & =\left. \frac d{d\epsilon}\right|_{\epsilon = 0} \F\big(\{\exp(-\epsilon \xi_{n+1}) L_n\}\big)+\left. \frac d{d\epsilon}\right|_{\epsilon = 0} \F\big( \{L_n\exp(\epsilon \xi_n\}) \big)\\
& =-\sum_{n=1}^N\mbox{tr}\big(\nabla_n\F\xi_{n+1}\big) + \sum_{n=1}^N\mbox{tr}\big(\nabla'_n\F\xi_{n}\big)   = \sum_{n=1}^N \mbox{tr}\big( (-\T^{-1}\nabla_n\F + \nabla'_n\F) \xi_n\big)
\end{split}
\]
must hold for any $\xi \in H^N$. It follows that $\T^{-1}\nabla_n \F - \nabla_n'\F$ must belong to the annihilator of the subalgebra $\h$, which is the dual of $\g_{-1}\oplus\g_0$, namely $\g_1$.
\end{proof}
\begin{remark}
Conditions \eqref{onecond}, \eqref{maincond} determine $\nabla\F$ uniquely  in terms of the geometric invariants and the variational derivatives $\delta_n f$. The proof of this fact can be found in \cite{MW}.
\end{remark}

Moreover, these conditions lead to a concrete representation of the reduction of the twisted Poisson bracket~\eqref{twisted} to the moduli space $G^N/H^N$:
\begin{equation}\label{reducedb}
 \{f, h\}_R(\a) : = \{\F, \H\}(L), 
\end{equation}
with $f, h$ smooth functionals on $G^N/H^N$, $\F, \H$ their invariant extensions to $G^N$, $\{, \}$ as in \eqref{twisted}, $\a=\{\a_n\}\in G^N/H^N$, and $L = \{L_n\}\in G^N$ as in \eqref{L}.

\begin{prop}\label{expl} Given $f,h,\F, \H$ as above, the reduced bracket \eqref{reducedb} has the following explicit representation:
\begin{equation}
\label{explicit}
\{f, h\}_R=\langle ( \nabla \F-\T  \nabla' \F)_1, ( \nabla \H )_{-1}\rangle=\sum_{n=1}^N  {\mathrm{tr}}\big( ( \nabla_n \F -\T \nabla'_n \F)_1( \nabla_n \H )_{-1}\big),
\end{equation}
where the subscripts indicate the components w.r.t. the gradation $\g=\g_{1}\oplus\g_0\oplus\g_{-1}$.
\end{prop}
\begin{proof}
Using \eqref{twisted}, with  $r(\xi\wedge \eta) =\frac12 \mbox{tr}(\xi (\eta_+-\eta_-))$  and $r(\xi\otimes\eta)=\mbox{tr}(\xi_-\eta_+)+\sum_k\mbox{tr}(\xi_0^k\eta_0^k)$, we compute
\[
\begin{split}
\{f, h\}_R = \frac12 &\sum_{n=1}^N\mbox{tr}\big( (\nabla_n \F)_{-} (\nabla_n \H)_{+}\big) -  \mbox{tr}\big( (\nabla_n \H)_{-} (\nabla_n \F)_{+}\big)\\  
+  \frac12 &\sum_{n=1}^N \mbox{tr}\big( (\nabla_n' \F)_{-}(\nabla_n' \H)_{+}\big) -  \mbox{tr}\big((\nabla'_n \H)_{-} (\nabla_n' \F)_{+}\big) \\
\  -& \sum_{n=1}^N \mbox{tr}\big( \T(\nabla_n' \F)_{-} (\nabla_n \H)_{+}\big)  + \mbox{tr}\big( \T(\nabla_n' \H)_{-} (\nabla_n \F)_{+}\big)\\
+&\sum_{n=1}^N-\mbox{tr}\big( \T(\nabla_n' \F)_{0}^{k} (\nabla_n \H)_{0}^k\big)+\mbox{tr}\big( \T(\nabla_n' \H)_{0}^{k} (\nabla_n \F)_{0}^k\big).
\end{split}
\]
Since $\g_1 \subset \g_+$, condition \eqref{maincond} implies that $\T^{-1}\nabla_n \F - \nabla_n'\F \in \g_+$ and therefore that $(\T\nabla'_n \F)_{-} = (\nabla_n\F)_-$ and $(\T\nabla'_n \F)_{0}^k = (\nabla_n\F)_0^k$ for all $k$. Similarly, $(\T\nabla'_n \H)_{-} = (\nabla_n\H)_-$ and $(\T\nabla'_n \H)_{0}^k = (\nabla_n\H)_0^k$. Using these relationships we get
\[
\begin{split}
\{f, h\}_R& =\frac12\sum_{n=1}^N\mbox{tr}\big( (\nabla_n \F)_{+} (\nabla_n \H)_{-}\big) -  \mbox{tr}\big( (\nabla_n \H)_{+} (\nabla_n \F)_{-}\big) \\
& -\mbox{tr}\big( (\nabla_n' \F)_{+}(\nabla_n' \H)_{-}\big)+  \mbox{tr}\big((\nabla'_n \H)_{+}(\nabla_n' \F)_{-}\big).  
\end{split}
\]
Using condition~\eqref{maincond} for $\H$, we have
\[
\begin{split}
\sum_{n=1}^N\mbox{tr}\big((\nabla_n \F)_{+} (\nabla_n \H)_{-}\big) & - \mbox{tr}\big((\nabla_n' \F)_{+} (\nabla_n' \H)_{-}\big)\\ 
& = \sum_{n=1}^N\mbox{tr}\big((\nabla_n \F)_{+} (\nabla_n \H)_{-}\big) - \mbox{tr}\big( (\nabla_n' \F)_{+} \T^{-1}(\nabla_n \H)_{-}\big)\\
& = \sum_{n=1}^N\mbox{tr}\big( (\nabla_n \F -\T\nabla'_n\F)_+, (\nabla_n \H)_{-}\big).
\end{split}
\]
Similarly,
\[
\begin{split}
\sum_{n=1}^N\mbox{tr}\big((\nabla_n \H)_{+} (\nabla_n \F)_{-}\big) & - \mbox{tr}\big((\nabla_n' \H)_{+} (\nabla_n' \F)_{-}\big)\\ 
& = \sum_{n=1}^N\mbox{tr}\big( (\nabla_n \H -\T\nabla'_n\H)_+(\nabla_n \F)_{-}\big).
\end{split}
\]
On the other hand, from the identity $\mbox{tr}(\xi_1 \xi_-)=\mbox{tr}(\xi_1\xi_{-1})$ and the fact that $\F$ and $\H$ satisfy~\eqref{maincond}, we rewrite the expression of the reduced bracket as
\[
\{f, h\}_R = \frac12 \sum_{n=1}^N\mbox{tr}\big( (\nabla_n \F -\T\nabla'_n\F)_1(\nabla_n \H)_{-1}\big)- \mbox{tr}\big( (\nabla_n \H -\T\nabla'_n\H)_1 (\nabla_n \F)_{-1}\big).
\]
Since $\nabla_n \H-\T\nabla'_n\H\in \g_1$, the second term becomes
\[
\begin{split}
&  \mbox{tr}\big( (\nabla_n \H -\T\nabla'_n\H)_1(\nabla_n \F)_{-1}\big) =\mbox{tr}\big( (\nabla_n \H -\T\nabla'_n\H) \nabla_n \F\big) \\  = & -\mbox{tr}\big( \nabla_n \F \, \T\nabla'_n\H\big) +\mbox{tr}\big(\nabla_n \H \, \nabla_n \F\big) = -\mbox{tr}\big( \nabla_n \F\, \T\nabla'_n\H\big) +\mbox{tr}\big(\nabla'_n \H \, \nabla'_n \F\big),
\end{split}
\]
where in the last step we used relation~\eqref{lrgrad} and the invariance of the trace under conjugation. It follows that the operator  $\nabla - \T \nabla'$  is skew-symmetric on the set of invariant extensions. In fact, shifting indices appropriately and using~\eqref{maincond},
\[
\begin{split}
&   \sum_{n=1}^N\mbox{tr}\big( (\nabla_n \H -\T\nabla'_n\H)_1 (\nabla_n \F)_{-1}\big) 
=\sum_{n=1}^N -\mbox{tr}\big( \nabla_n \F\, \T\nabla'_n\H\big) +\mbox{tr}\big( \nabla'_n \H  \, \nabla'_n \F\big)\\
 = & \sum_{n=1}^N -\mbox{tr}\big( (\nabla_n \F-\T \nabla'_n \F)\, \T\nabla'_n\H\big)=-\sum_{n=1}^N \mbox{tr}\big( (\nabla_n \F -\T \nabla'_n \F)_1\, (\T\nabla'_n\H)_{-1}\big) \\  =& -\sum_{n=1}^N \mbox{tr}\big( (\nabla_n \F -\T \nabla'_n \F)_1\, (\nabla_n\H)_{-1}\big).
\end{split}
 \]
This finally gives 
 \[
\{f, h\}_R =  \sum_{n=1}^N\mbox{tr}\big( (\nabla_n \F -\T\nabla'_n\F)_1\, (\nabla_n \H)_{-1}\big).
\]
\end{proof}

The companion Poisson bracket, 
obtained in \cite{MW} as the reduction of the right tensor $\theta$ to the moduli space, also has the following explicit representation:
\begin{equation}
\label{reduced0}
\{f, h\}_0=\frac12\sum_{n=1}^N \mbox{tr}\big((\nabla_n' \F)_{-1} \,(\nabla_n' \H)_{1})\big)-\mbox{tr}\big((\nabla_n' \H)_{-1}\, (\nabla_n' \F)_{1})\big) .
\end{equation}

\subsubsection{Invariant Hamiltonian flows}
 
Next we study the relation between Hamiltonian flows induced by the reduced bracket $\{ \, , \}_R$ on the geometric invariants ${a_n^j}$'s, viewed as coordinates of the moduli space of $\M^1_N$, and the associated polygonal evolutions.

Assume that $\a=\{\a_n\}$ evolves according to the Hamiltonian evolution
\begin{equation}\label{ahamflow}
\frac{\rd}{\rd t}{\a}_n=\{ f, \a_n\}_R,
\end{equation}
for given Hamiltonian $f:G^N/H^N \rightarrow \RR$.  The following key result relates the matrix $Q=\{Q_n\}$ defining both polygon and frame evolution to the right and left gradients of an invariant extension $\F$ of $f$.

\begin{lem}
Let $\gam=\{ \gam_n \}\in \M^1_N$ and  $Q=\{Q_n\}$ as in~\eqref{Q}. Then,  the polygon evolution~\eqref{pev} (equivalently, the evolution~\eqref{frameev} of its left moving frame) induces the Hamiltonian equation~\eqref{ahamflow}  if and only if
\begin{equation}\label{gradLeftQcomp}
(Q_{n+1})_{-1} = (\nabla_n \F)_{-1}.
\end{equation}
\end{lem}

\begin{proof}
Given an arbitrary smooth functional $h:G^N/H^N \rightarrow \RR$, we use the following form the compatibility condition~\eqref{compat} 
\begin{equation}\label{determineQ}
\begin{pmatrix}\0_{m-1}&\ddt \a_n\\ \0^T&0\end{pmatrix} = Q_{n+1}-K_n^{-1}Q_nK_n = Q_{n+1}-L_nQ_nL_n^{-1},
\end{equation}
and expression~\eqref{onecond} of the left gradient of $h$'s invariant extension $\H$ 
to write
\[
\ddt h(\a)=\sum_{n=1}^N\delta_n h^T \cdot \ddt\a_n =\langle LQL^{-1}-\T Q , \nabla\H\rangle,
\]
If the $\a_n$'s evolve according to~\eqref{ahamflow}, using the explicit representation~\eqref{explicit} of the reduced bracket, we also have
\[
\ddt h(\a)=\{f, h\}_R=\langle ( \nabla \F-\T  \nabla' \F)_1, ( \nabla \H )_{-1}\rangle.
\]
Rewriting
\[
\begin{split}
&\langle LQL^{-1}-\T Q , \nabla\H \rangle 
=\langle LQL^{-1}, \nabla\H\rangle - \langle \T Q , \nabla\H\rangle 
=\langle Q, L^{-1} \nabla \H L\rangle  -\langle \T Q, \nabla\H\rangle \\ &=\langle Q,\nabla' \H\rangle  -\langle \T Q , \nabla\H\rangle =\langle \T Q,\T\nabla' \H\rangle -\langle \T Q , \nabla\H\rangle
 = \langle \T Q , \T\nabla' \H-\nabla\H\rangle,
\end{split}
\]
and 
\[
\langle ( \nabla \F-\T  \nabla' \F)_1, ( \nabla \H )_{-1}\rangle=-\langle ( \nabla \H-\T  \nabla' \H)_1, ( \nabla \F )_{-1}\rangle=
\langle \nabla \F ,  \T  \nabla' \H -\nabla \H\rangle,
\]
and comparing the two expressions for $\ddt h(\a)$, we obtain
\[
\langle \T Q -\nabla \F, \T\nabla' \H-\nabla\H\rangle = 0.
\]
Since $\T\nabla_n' \H-\nabla_n\H \in \g_1$ for any invariant $\H$, then $(Q_{n+1}-\nabla_n \F)_{-1}$ must vanish and~\eqref{gradLeftQcomp} follows.
 \end{proof}

 \begin{remark}
The right gradient of $\F$ satisfies the analogous relation 
 \begin{equation}\label{gradRightQcomp}
 (Q_n)_{-1} = (\nabla_n' \F)_{-1},
 \end{equation}
derived directly from using condition \eqref{maincond} in~\eqref{gradLeftQcomp}.
\end{remark}

Since 
\begin{equation}\label{QfromX}
\begin{split}
Q_n& =\rho_n^{-1}\begin{pmatrix} X_n & X_{n+1} \ldots & X_{n+m-1} \end{pmatrix}=\rho_n^{-1}\begin{pmatrix} \rho_n \r_n &\rho_{n+1} \r_{n+1}& \ldots &  \rho_{n+m-1} \r_{n+m-1} \end{pmatrix} \\
&=\begin{pmatrix} \r_n & \rho_n^{-1} \rho_{n+1}\r_{n+1} & \ldots &\rho_n^{-1} \rho_{n+m-1} \r_{n+m-1} \end{pmatrix},
\end{split}
\end{equation}
where the vector field $X_n =  \rho_n \r, \, \r^T = (r_n^0, \dots , r_n^{m-1})$ defines the polygon evolution ~\eqref{pev}, the entries $r^i_n,\, i=1, \ldots n+m-1$  of $\r$ are uniquely determined by  $(Q_n)_{-1}$,  while the remaining entry $r^0_n$ is determined by the arc length preservation condition as in Proposition~\ref{LAP}. From this observation and relation~\eqref{gradLeftQcomp} it follows that:
 \begin{theorem}\label{Xf} Given a smooth function $f:G^N/H^N \rightarrow \RR$, there exists an invariant vector field $X^f$ on $\M^1_N$,  constructed as above,  that induces the evolution~\eqref{ahamflow} on the geometric invariants $\{\a_n\}$ with Hamiltonian $f$. \end{theorem}
\begin{remark} Adding to $X^f$ any vector field for which $(\a_n)_t = 0$ (i.e.~an  infinitesimal rigid symmetry) will induce the same evolution~\eqref{ahamflow} on the geometric invariants. (The infinitesimal symmetries will be later shown to form the kernel of a pre-symplectic form $\omega_1$ introduced in Section~\ref{sympSec}.) However, the construction leading to the vector field described in Theorem \ref{Xf}  determines $X^f$ uniquely. In fact, for given $f$,  the gradient of an extension $\F$ is uniquely defined by the condition $\T\nabla_n' \F-\nabla_n\F \in \g_1$. Then,  one constructs $X^f$ as the linear combination of the columns of the moving frame $\rho_n$ with coefficients given by the entries of $(Q_n)_{-1} = (\nabla_n' \F)_{-1}$ (in other words, by the entries of the first column of $Q_n$). This construction is discussed in more detail in \cite{MW}. \end{remark}

The $\g_1$-component of the matrix $Q:=Q^f$, associated with this choice of vector field $X^f$, can be identified with the vector of variational derivatives of the Hamiltonian $f$. This result will be used in Section~\ref{sympSec} to relate the reduced brackets to a pair of pre-symplectic forms on $\M^1_N$.

\begin{proposition}\label{Qtovariations}
Let $X^f$ be the vector field described in Theorem~\ref{Xf} and $Q^f$  be as in~\eqref{QfromX}.  Then
\begin{equation}\label{deltaform}
(Q_n^f)_1 = \begin{pmatrix} 0& (-1)^m \delta_n f^T\\ \0 & \0_{m-1}\end{pmatrix}.
\end{equation}
\end{proposition}
%
\begin{proof} 
Writing $Q$ in the following form
\begin{equation}\label{D}
Q_n^f = \begin{pmatrix}D_n^{0} & D_n^{1} & \dots  & D_n^{m-1}\end{pmatrix},
\end{equation}
with $D_n^0=\r_n$, $D_n^i = K_n\dots K_{n+i-1}\r_{n+i}$, $i=1,\dots m-1$, one has 
\begin{equation}
\label{relations}
K_n D_{n+1}^i = D_n^{i+1}, \qquad i=0, \ldots m-1.
\end{equation}
We will also make use of the following easily verifiable identities:
\[
L_n\e_k = \e_{k-1}, \ k=2, \ldots, m,   \qquad \quad  \e^T_m L_n=(-1)^{m-1}\e_1^T.
\]

First, it is easy to verify by direct computation that components $2, \ldots m$ of the top row of $K_n\nabla_n \F K_n^{-1}$ are the components of $(-1)^m \delta_n f^T$, i.e.
\[
\e^T_1K_n\nabla_n \F K_n^{-1} \e_s=(-1)^m \frac{\delta f}{\delta a_n^{s-1}}, \quad s=2, \ldots m.
\]
Next, relation~\eqref{gradLeftQcomp}: $(Q^f_{n+1})_{-1} = (\nabla_n \F)_{-1}$, can be written as
\[
\nabla_n \F \e_1=D^0_{n+1} + c_0\, \e_1,
\]
for some $c_0\in \mathbb{R}$. Multiplying the above on the left by $K_n$ and using~\eqref{relations}, we get
\begin{equation}
\label{nstep}
K_n\nabla_n \F \e_1=K_nD^0_{n+1} + c_0 K_n\e_1=D^1_n+c_0\, \e_2.
\end{equation}
It follows that
\[
\e^T_1 D^1_n=\e^T_1 K_n\nabla_n \F K_n^{-1} K_n \e_1=\e^T_1 K_n\nabla_n \F K_n^{-1} \e_2=(-1)^m \frac{\delta f}{\delta a_n^{1}}.
\]
Shifting indices in~\eqref{nstep}, we write
\[
K_{n+1}\nabla_{n+1} \F \e_1=D^1_{n+1}+c_0\, \e_2.
\]
Its left-hand side can be rewritten as 
\[
K_{n+1}\nabla_{n+1} \F K_{n+1}^{-1} K_{n+1}\e_1=K_{n+1}\nabla_{n+1} \F K_{n+1}^{-1}\e_2=\nabla_n \F \e_2 + c_1\e_2,
\]
for $c_1\in \mathbb{R}$, where we used relation~\eqref{lrgrad} between the right and left gradient combined with condition~\eqref{maincond} to give $(K_{n+1}\nabla_{n+1} \F K_{n+1}^{-1})_1=(\nabla_{n}\F)_1$. Multiplying by $K_n$ on the left, we arrive at
\begin{equation}
\label{n+1step}
K_n\nabla_n \F \e_2+c_1\e_3=K_n D^1_{n+1}+c_0\, \e_3=D^2_n + c_0\, \e_3.
\end{equation}
Then,
\[
\e^T_1 D^2_n=\e^T_1 K_n\nabla_n \F K_n^{-1} K_n \e_2=\e^T_1 K_n\nabla_n \F K_n^{-1} \e_3=(-1)^m \frac{\delta f}{\delta a_n^{2}}.
\]
The claim follows by induction.

%
\end{proof}

Using the results in this section, one can rewrite formula (\ref{explicit}) for the reduced Poisson bracket as follows:
\begin{equation}\label{explicit2}
\begin{split}
 \{f, h\}_R(\a)   =  & -\frac12\sum_{n=1}^N\sum_{r=1}^{m-1}\left(\e_r^T \left(L_n Q^f_n L_n^{-1}-\T Q^f_n\right) \e_m \frac{\delta h}{\delta a_n^r} \right. \\ &- \left . \e_r^T\left(L_n Q^h_n L_n^{-1}-\T Q^h_n\right)\e_m \frac{\delta f}{\delta a_n^r} \right).
 \end{split}
\end{equation}


%
%
%
%
\begin{remark}
The Hamiltonian equation~\eqref{ahamflow} can be written as
\begin{equation}\label{hamflowP}
\ddt \a_n=\Pc_{1n} \delta_n f,
\end{equation}
for a skew-symmetric operator $\Pc_{1n}$ explicitly computable from the ``atomic" brackets $\{ a^s_n, a^r_n\}_R$. 
\end{remark}

\begin{corollary}\label{a-p}  Given $X\in \V^1_N$, let $Q^X$ define the associated frame evolution~\eqref{frameev}. Write
\begin{equation}\label{peq}
(Q_n^X)_1 = \begin{pmatrix} 0& (-1)^m \p_n^T\ \\ 0 & \0_{m-1}\end{pmatrix}.
\end{equation}
Then
\begin{equation}\label{phamil}
\ddt \a_n = \Pc_n \p_n.
\end{equation}
\end{corollary}
\begin{proof} The statement is true if  $X=X^f$ comes from a Hamiltonian $f$ as in Theorem~\ref{Xf}. In this case, Proposition~\ref{Qtovariations} and equation~\eqref{hamflowP} give $\p=\delta_nf$.

For the general case, $(Q)_1$ together with the compatibility condition \eqref{compat}) completely determine $Q$. Since the construction is algebraic, given~\eqref{peq}, we can simply replace $\delta_n f$ in \eqref{hamflowP} with $\p_n$.
\end{proof}

\begin{example}[$m=3$] For sake of readability, we drop the subscript $n$. Given expression~\eqref{deltaform} for the $\g_1$-component of $Q^f$, we write:
\[
Q^f = \begin{pmatrix} a& -\delta_k f & -\delta_\tau f\\ c&-a-b&e\\ h & g & b\end{pmatrix},
\]
where $(k, \tau)^T$ is the vector of invariants, $\delta_k f = {\delta f}/{\delta k_n}$, and $\delta_\tau f = {\delta f}/{\delta \tau_n}$. With $\small L= \begin{pmatrix} -k & 1 & 0 \\ -\tau & 0 & 1\\ 1 & 0 & 0\end{pmatrix}$,
\[
LQ^f L^{-1}=\begin{pmatrix}k\delta_k f-a-b&e+k\delta_\tau f&-ka+c+k(k\delta_k f-a-b)+\tau(k\delta_\tau f+e)\\g+\tau\delta_k f&b+\tau\delta_\tau f&-a\tau+h+(\tau\delta_k f+g)k+\tau(\tau\delta_\tau f+b)\\-\delta_k f&-\delta_\tau f&-k\delta_k f-\tau \delta_\tau f+a\end{pmatrix}.
\]
Since $\T Q^f-\L Q^fL^{-1}$ is determined by~\eqref{determineQ}, we set all but components $(3,1)$ and $(3,2)$ to zero, getting:
\[
\begin{split}
a &=\Rc^{-1}\left((1+\T)k\delta_k f+\tau\delta_\tau f\right), \quad b = -\Rc^{-1}\left(\T k\delta_k f+(1+\T)\tau \delta_\tau f\right),\\
\T h & = -\delta_k f, \quad \T g = -\delta_\tau f, \quad
\T c  =\tau\delta_k f-\T^{-1}\delta_\tau f, \quad e+k\delta_\tau f= -\T^{-1}\delta_k f
\end{split}
\]
where $\Rc = 1+\T+\T^2$. 
Using formula~\eqref{explicit2}, we compute the reduced bracket
\[
\{f, h\}_R =  \sum_{n=1}^N \begin{pmatrix}\delta_k f&\delta_\tau f\end{pmatrix}\Pc_1\begin{pmatrix} \delta_k h\\\delta_\tau h\end{pmatrix},
\]
where 
\begin{equation}\label{P3}
\Pc_1= \begin{pmatrix} \T^{-1}\tau-\tau\T +k\Rc^{-1}(\T^2-1)k& \T-\T^{-2}+k\Rc^{-1}(\T-1) \tau \\\\ \T^2-\T^{-1}+\tau\Rc^{-1}\T(\T-1)k&\T k-k\T^{-1}+\tau\Rc^{-1}(\T^2-1)\tau
 \end{pmatrix},
\end{equation}
is the same tensor as in (\ref{Poisson3}).

The companion bracket found in \cite{MW}: 
\[
\{f,h\}_0(k,\tau): =\theta(\F,\H)(L) = \sum_{n=1}^N r(\nabla'_n \F\wedge \nabla'_n\H),
\]
can be similarly computed as
\[
\{f,h\}_0(k,\tau) =  \sum_{n=1}^N \begin{pmatrix}\delta_k f&\delta_\tau f\end{pmatrix}\Pc_2\begin{pmatrix} \delta_k h\\\delta_\tau h\end{pmatrix},
\]
where
\begin{equation}\label{P30}
\Pc_2 = \begin{pmatrix} \T^{-1}\tau-\tau\T&\T-\T^{-2}\\ \T^2-\T^{-1}&0\end{pmatrix}.
\end{equation}
The skew-symmetric operators $\Pc_1$ and $\Pc_2$  define compatible Poisson structures and form a Hamiltonian pencil for the second Adler-Gelfand-Dikii flow, an integrable discretization of the Boussinesq equation~\cite{MW}.

\end{example}
\section{Lifting Hamiltonian structures to the moduli space of twisted polygons}\label{sympSec}
In higher dimension, the counterparts of brackets $\{\ , \, \}_R, \{\ , \, \}_0$ are also  Poisson brackets on the space of geometric invariants, as shown in~\cite{MW}. However,  a proof of their compatibility beyond dimension 2 could not be attained: on the one hand,   tensor $\theta$ (see~\eqref{theta}) is neither a Poisson tensor nor compatible with the twisted bracket~\eqref{twisted}; on the other hand, the reduction process becomes too complicated to allow a direct or a theoretical argument. 
As discussed in this section, lifting the brackets to the moduli space $\M^1_N$ produces a far simpler pair of pre-symplectic forms on $\V_N^1$, in terms of which proof of compatibility becomes elementary. We first illustrate the main ideas by continuing the discussion of the 3-dimensional case.
\subsection{Example ($m=3$)} Let $X,Y\in \V_N^1$ be arc length-preserving left invariant vector fields. At a given $\gam \in \M^1_N$,  and write their components as
 \[
X_n = c_n \gam_n + b_n \gam_{n+1} + a_n\gam_{n+2}, \qquad Y_n = \hc_n \gam_n + \hb_n \gam_{n+1} + \ha_n\gam_{n+2},
\]
with $a_n,  \ha_n, b_n, \hb_n, c_n, \hc_n$ periodic functions of the invariants $\{k_n\},\{\tau_n\}$ satisfying  the arc length preservation condition~\eqref{c3}. 
\smallskip


For  $Z\in \V^1_N$, let $J_n(Z) := Z_n + k_nZ_{n+1}+\tau_nZ_{n+2}$, and define the following $2$-forms on $\M_N^1$:
\[
\begin{split}
\omega_1(X,Y)&= \sum_{n=1}^N|J_n(Y),X_{n+1},\gam_{n+2}|+|J_n(Y), \gam_{n+1},X_{n+2}| \\ & - |J_n(X),Y_{n+1},\gam_{n+2}| 
-|J_n(X), \gam_{n+1},Y_{n+2}|\\ &+ |X_n,\gam_{n+1},Y_{n+2}|+|\gam_n,X_{n+1},Y_{n+2}| - |Y_n,\gam_{n+1},X_{n+2}|-|\gam_n,Y_{n+1},X_{n+2}|,\\
\omega_2(X, Y)&= \sum_{n=1}^N |Y_n,X_{n+1},\gam_{n+2}|+|Y_n,\gam_{n+1},X_{n+2}| - |X_n,Y_{n+1},\gam_{n+2}|-|X_n,\gam_{n+1},Y_{n+2}|,
\end{split}
\]
\begin{theorem}\label{omega3} Let  $\mathcal{P}_{1}$, $\mathcal{P}_{2}$ be the tensors defined by~\eqref{P3},~\eqref{P30}.  Denote with $Q^X$ (resp. $Q^Y$) the matrices associated with vector field $X$ (resp. $Y$) as in~\eqref{Q3}, and write the $\g_{1}$-components as
\[
Q^X_n = \begin{pmatrix}\ast&-\p_n^T\\ \ast&\ast\end{pmatrix}, \quad\quad Q^Y_n = \begin{pmatrix}\ast&-\hat\p_n^T\\ \ast&\ast\end{pmatrix},
\]
with $\p_n, \hat{\p}_n \in \RR^2.$
Then
\begin{equation}\label{pomegas}
\omega_1(X,Y) = \sum_{n=1}^N \p_n^T\Pc_{1n} \hat\p_n, \quad\quad \omega_2(X,Y) = \sum_{n=1}^N \p_n^T\Pc_{2n}\hat\p_n.
\end{equation}
\end{theorem}
\begin{proof} 
We prove the result for $\omega_2$, the simpler of the 2-forms. Using the expression $Q_n= \begin{pmatrix} \r_n & K_n \r_{n+1} & K_n K_{n+1} \r_{n+2}\end{pmatrix}$, with $\r_n=(c_n \ b_n \ a_n)^T$ and $K_n$ given by~\eqref{K2D},
\[
-\p_n = \begin{pmatrix} \e_1^TK_n\r_{n+1}\\ \e_1^TK_nK_{n+1}\r_{n+2}\end{pmatrix} = \begin{pmatrix}a_{n+1}\\ b_{n+2}+\tau_{n+1}a_{n+2}\end{pmatrix}.
\]
Setting $v_n=a_{n+1}$, $w_n=b_{n+2}+\tau_{n+1}a_{n+2}$ (as in Section~\ref{3devol}), we get $\p_n=-(v_n, w_n)^T$.

Next, using the fact that $\rd \rho_n/\rd t = (X_n, X_{n+1}, X_{n+2}) = \rho_nQ_n$, we compute
\[
\begin{split}
&|Y_{n}, X_{n+1}, \gam_{n+2}|+|Y_n,\gam_{n+1},X_{n+2}| =|\rho_nQ^Y_n\e_1, \rho_nQ^X_n\e_2, \rho_n\e_3|+|\rho_nQ^Y_n\e_1, \rho_n\e_2, \rho_nQ_n^X\e_3|\\
&=|Q^Y_n\e_1, Q^X_n\e_2, \e_3|+|Q^Y_n\e_1, \e_2, Q_n^X\e_3|=|\hr_n, K_n\r_{n+1},\e_3| + |\hr_{n},\e_2,K_nK_{n+1}\r_{n+2}|\\
&=\big(\hc_n\e_2^TK_n\r_{n+1}-\hb_n\e_1^TK_n\r_{n+1}\big)+\big(\hc_n\e_3^TK_nK_{n+1}\r_{n+2}-\ha_n\e_1^TK_nK_{n+1}\r_{n+2}\big)\\
&=\hc_n\big(c_{n+1}+c_{n+2}+(1+\T)k_nv_n+\tau_nw_n\big)-v_n\hw_{n-2}-\hv_{n-1}w_n+\tau_{n-1}v_n\hv_{n-1}\\
&=-\hc_n c_n +\begin{pmatrix}v_n & w_n\end{pmatrix}^T \begin{pmatrix}   \T^{-1} \tau_{n} & -\T^{-2} \\ -\T^{-1}& 0       \end{pmatrix} \begin{pmatrix} \hv_n \\ \hw_n\end{pmatrix},
\end{split}
\]
where, in the last step, we used the arc length preservation condition~\eqref{c3} to write
\[
c_{n+1} + c_{n+2}+(1+\T)k_nv_n+\tau_n w_n=-c_n.
\]
Swapping $X$ and $Y$ in the computation above amounts to interchanging the hat symbol; subtracting the resulting expression, we get:
\[
\begin{pmatrix}v_n&w_n\end{pmatrix}\begin{pmatrix}\T^{-1}\tau_n-\tau_n\T&\T-\T^{-2}\\ \T^2-\T^{-1}&0\end{pmatrix}\begin{pmatrix}\hv_n\\\hw_n\end{pmatrix}=\begin{pmatrix}v_n&w_n\end{pmatrix}\Pc_{2n}\begin{pmatrix}\hv_n\\\hw_n\end{pmatrix}.
\]
Summing over $n$ completes the proof for $\omega_2$. The formula for $\omega_1$ can be obtained in a similar way, with 
a similar, slightly more involved computation. 
\end{proof}

Let $h$ be a smooth functional on $G^N/H^N$ and let $Y=X^h$ be the vector field in $\V^1_N$ uniquely constructed as in Theorem~\ref{Xf}. Then, using~\eqref{deltaform}, we replace  $\hat\p_n=\delta_n h$ and obtain
\[
\omega_1(X,X^h) = \sum_{n=1}^N  \p_n^T\Pc_{1n} \delta_nh.
\]
On the other hand, assume that $\rd \gamma_n/\rd t= X_n$, and let $\pi:  \M_N^1  \to  G^N/H^N$, $\pi(\gam):=\a$,
be the projection of $\M^N_1$ to the moduli space. Then,
\[
\rd(h\circ\pi)(X) = X(h\circ\pi) = \ddt h(\a\circ \gam)=\sum_{n=1}^N \delta_n h^T \ddt \a_n.
\]
Using \eqref{ktev3} to replace $\rd\a_n/{\rd t} = -\Pc_{1n} \p_n$, we write
\[
\sum_{n=1}^N \delta_n h^T \ddt \a_n=-\sum_{n=1}^N \delta_n h^T \Pc_{1n }\p_n=\sum_{n=1}^N\p_n^T \Pc_{1n} \delta_n h,
\]
finally giving
\begin{equation}\label{diffsymp}
\rd(h\circ\pi)(X)=\omega_1(X,X^h), \qquad \forall X\in \V^1_N.
\end{equation}
\begin{remark}
We will show in the general case that $\omega_1$ is a pre-symplectic form, i.e.~a closed 2-form with non-trivial kernel, thus Hamiltonian vector fields are only defined up to addition of elements in the kernel. Nevertheless, it is its connection to $\omega_2$ that selects $X^h$ as the natural representative of the class of Hamiltonian vector fields associated to $h$ via $\omega_1$.
\end{remark}

\subsection{The general case.}
We are now ready to discuss our main results.  Given $X, Y\in \V_N^1$, we introduce the following pair of 2-forms on $\M_N^1$:
%
  \begin{equation}\label{omega1} 
\begin{split}
 \omega_1(X,Y)(\gam) := & \frac{1}{2}\sum_{n=1}^N\sum_{r=1}^{m-1}\Big(|J_n(Y), \gam_{n+1}, \dots, \gam_{n+r-1},\stackrel{r+1}{\overbrace{X_{n+r}}}, \gam_{n+r+1}, \dots, \gam_{n+m-1}| \\ & - |J_n(X), \gam_{n+1}, \dots, \gam_{n+r-1},\stackrel{r+1}{\overbrace{Y_{n+r}}}, \gam_{n+r+1}, \dots, \gam_{n+m-1}| \\
&   -|\gam_n,\dots,\stackrel{r}{\overbrace{X_{n+r-1}}},\dots,Y_{n+m-1}|+ |\gam_n,\dots,\stackrel{r}{\overbrace{Y_{n+r-1}}},\dots,X_{n+m-1}|\Big),
\end{split}
\end{equation}
\smallskip

\noindent
where $\displaystyle J_n(X) = (-1)^{m-1}\sum_{\ell=0}^{m-1} a_n^\ell X_{n+\ell}$, with $a_n^0 = (-1)^{m-1}$;

\begin{equation}\label{omega2}
\begin{split}
 \omega_2(X, Y) (\gam): = & \frac 12\sum_{n=1}^N\sum_{r=1}^{m-1}\Big(|Y_{n}, \gamma_{n+1},\dots, \gam_{n+r-1},\stackrel{r+1}{\overbrace{X_{n+r}}},\gam_{n+r+1} \dots,\gamma_{n+m-1}| \\
&   - |X_{n}, \gamma_{n+1},\dots, \gam_{n+r-1},\stackrel{r+1}{\overbrace{Y_{n+r}}},\gam_{n+r+1} \dots,\gamma_{n+m-1}|\Big).
\end{split}
\end{equation}


\begin{proposition}\label{domega2}
Given the 1-form 
\begin{equation}\label{theta2}
\theta_2(X) = \frac{1}{2}\sum_{n=1}^N |X_n, \gamma_{n+1}, \dots, \gamma_{n+m-1}|,
\end{equation}
defined on $\M_N^1$, then $\omega_2=\rd\theta_2.$
\end{proposition}
\begin{proof} 
The verification of this fact is straightforward:
\[
d\theta_2(X,Y) = X(\theta_2(Y))-Y(\theta_2(X)) - \theta_2([X,Y]) 
\]
\[
=  \frac{1}{2}\sum_{n=1}^N \left|[X,Y]_n,\gamma_{n+1}, \dots, \gamma_{n+m-1}\right|+ \omega_2 (X,Y) -\theta_2([X,Y]) = \omega_2 (X,Y) .
\]
\end{proof}

\begin{cor}\label{2closed} The $2$-form $\omega_2$ is a closed form on $\M_N^1.$
\end{cor}
\begin{remark}
Also $\omega_1$ can be rewritten in terms of the 1-form $\theta_2$. If we introduce the operator $L(X) = \{L_n(X)\}$ with 
\[
L_n(X) = (-1)^{m}\left((-1)^m X_n + a_n^1 X_{n+1}+\dots+a_n^mX_{n+m} - X_{n+m+1}\right).
\]  
Then,
\begin{equation}
\label{omega2alt}
\omega_1(X, Y)  =  X(\theta_2(L(Y))) - Y(\theta_2(L(X))) - \theta_2\left(\left(XL(Y) - YL(X)\right)\right).
\end{equation}
\end{remark}

\begin{theorem}\label{omega2br} Let $f, h$ be  smooth functionals on $G^N/H^N$ and let $X^f, X^h$ be the vector fields in $\V^1_N$ given by Theorem~\ref{Xf}. Then,
\[
 \omega_2(X^f,X^h)(\gam) = \{f, h\}_0(\a(\gam)),
\]
where $\{\ , \ \}_0$ is the reduction of the right tensor defined in~\eqref{reduced0}.
\end{theorem}
\begin{proof} 
In formula~\eqref{reduced0}, we use~\eqref{gradRightQcomp} to replace $(\nabla_n'\F)_{-1}$ with $(Q_{n}^f)_{-1}$, and  \eqref{onecomponent} together with~\eqref{deltaform} in Proposition~\ref{Qtovariations} to write $(\nabla_n'\F)_{1}=(Q_{n}^f)_{1}$, and do the same for $(\nabla_n'\H)_{\pm 1}$.
Writing $Q_n^f = \begin{pmatrix}D_n^{0} & D_n^{1} & \dots  & D_n^{m-1}\end{pmatrix}$,
$Q_n^h = \begin{pmatrix}R_n^{0} & R_n^{1} & \dots  & R_n^{m-1}\end{pmatrix}$, and denoting with $(D^i_n)_j$ (resp. $(R^i_n)_j$), $j=0, \ldots, m-1$, the $j$-th component of the $i$-column of $Q^f_n$ (resp. $R^f_n$), we compute:
\[
 \{f,h\}_0(\a) = (-1)^m\frac12\sum_{n=1}^N\sum_{k=1}^{m-1}(D_n^0)_k  \delta_n^k h- (R_n^0)_k  \delta_n^k f = \frac{1}2\sum_{n=1}^N\sum_{k=1}^{m-1}(D_n^0)_k (R_n^k)_0- (R_n^0)_k (D_n^k)_0.
 \]
On the other hand, $X^f_{n+r}=\rho_n D^r_n$ and $X^h_{n+r}=\rho_n R^r_n$. Replacing these expression in formula~\eqref{omega2} for $\omega_2$, we get
\[
\begin{split}
 \omega_2(X, Y) (\gam)= & \frac 12\sum_{n=1}^N\sum_{r=1}^{m-1}\Big(|\rho_n R_n^0, \gamma_{n+1},\dots, \gam_{n+r-1},\stackrel{r+1}{\overbrace{\rho_n D^r_n}},\gam_{n+r+1} \dots,\gamma_{n+m-1}| \\
&   - |\rho_n D^0_n, \gamma_{n+1},\dots, \gam_{n+r-1},\stackrel{r+1}{\overbrace{\rho_n R^r_n}},\gam_{m+r+1} \dots,\gamma_{n+m-1}|\Big).
\end{split}
\]
Computing, e.g., the first of the two determinants, we get
\[
\begin{split}
& |\rho_n R_n^0, \gamma_{n+1},\dots, \gam_{n+r-1},\stackrel{r+1}{\overbrace{\rho_n D^r_n}},\gam_{n+r+1} \dots,\gamma_{n+m-1}| \\
&=\sum_{j, k=0}^{m-1} |(R_n^0)_j\gam_j, \gamma_{n+1},\dots, \gam_{n+r-1},\stackrel{r+1}{\overbrace{ (D^r_n)_k \gam_k}},\gam_{n+r+1} \dots,\gamma_{n+m-1}|= (R_n^0)_0 (D^r_n)_{r}-(R_n^0)_r (D^r_n)_{0}.
\end{split}
\]
Finally, 
\[
\begin{split}
\omega_2(X, Y) (\gam)& =\frac 12\sum_{n=1}^N\left( (R_n^0)_0 \sum_{r=1}^{m-1} (D^r_n)_{r}-(D_n^0)_0 \sum_{r=1}^{m-1} (R^r_n)_{r}\right)+\frac{1}2\sum_{n=1}^N\sum_{r=1}^{m-1}(D_n^0)_r (R_n^r)_0- (R_n^0)_r (D_n^r)_0\\
&=\frac{1}2\sum_{n=1}^N\sum_{r=1}^{m-1}(D_n^0)_r (R_n^r)_0- (R_n^0)_r (D_n^r)_0,
\end{split}
\]
since $\mbox{tr}(Q^f)=\mbox{tr}(Q^h)=0$ implies
\[
(R_n^0)_0 \sum_{r=1}^{m-1} (D^r_n)_{r}-(D_n^0)_0\sum_{r=1}^{m-1} (R^r_n)_{r}=-(R_n^0)_0(D_n^0)_0+(D_n^0)_0(R_n^0)_0=0.
\] 
 \end{proof}

 \begin{theorem}\label{redth} Let $f, h$ be  smooth functionals on $G^N/H^N$ and let $X^f, X^h$ be the vector fields in $\V^1_N$ uniquely constructed as in Theorem~\ref{Xf}. Then
\begin{equation}\label{omega1br}
\omega_1(X^f, X^h)(\gam) = \{f, h\}_R(\a(\gamma)).
\end{equation}
\end{theorem}
\begin{proof}
Recall formula ~\eqref{omega1} for $\omega_1$:
\[
\begin{split}
 \omega_1(X,Y)(\gam) := & \frac{1}{2}\sum_{n=1}^N\sum_{r=1}^{m-1}\Big(|J_n(Y), \gam_{n+1}, \dots, \gam_{n+r-1},\stackrel{r+1}{\overbrace{X_{n+r}}}, \gam_{n+r+1}, \dots, \gam_{n+m-1}| \\ & - |J_n(X), \gam_{n+1}, \dots, \gam_{n+r-1},\stackrel{r+1}{\overbrace{Y_{n+r}}}, \gam_{n+r+1}, \dots, \gam_{n+m-1}| \\
&   -|\gam_n,\dots,\stackrel{r}{\overbrace{X_{n+r-1}}},\dots,Y_{n+m-1}|+ |\gam_n,\dots,\stackrel{r}{\overbrace{Y_{n+r-1}}},\dots,X_{n+m-1}|\Big),
\end{split}
\]
We first shift indices in last term:
  \[
 \sum_n \sum_\ell \left| \gamma_n, \dots, Y_{n+\ell-1}, \dots, X_{n+m-1}\right| = \sum_n \sum_\ell \left| \gamma_{n+1}, \dots, Y_{n+\ell}, \dots, X_{n+m}\right|.
 \]
Substituting
\[
\begin{split}
X_{n+m} &= \frac{\rd}{\rd t}\gamma_{n+m} = \frac{\rd}{\rd t}\sum_{k=0}^{m-1} a_n^k \gamma_{n+k} = \sum_{k=0}^{m-1} a_n^k X_{n+k} + \sum_{k=0}^{m-1} a_n^k (a_n^k)_t \gamma_{n+k}\\
& = (-1)^{m-1} J_n(X) + \sum_{k=0}^{m-1} (a_n^k)_t \gamma_{n+k}
\end{split}
\]
in the above, we obtain
\[
\begin{split}
& \sum_n \sum_\ell \left| \gamma_n, \dots, Y_{n+\ell-1}, \dots, X_{n+m-1}\right| \\ & =
 \sum_n \sum_\ell \left| J_n(X), \gamma_{n+1}, \dots, Y_{n+\ell}, \dots, \gamma_{n+m-1}\right|+ \sum_n \sum_\ell \sum_{k=1}^{m-1} (a_n^k)_t \left|\gamma_{n+1}, \dots, Y_{n+\ell}, \dots, \gamma_{n+k}\right|.
 \end{split}
 \]

The first term of the formula above cancels the second term in the expression for $\omega_1$. Also, the second term in the above can be rewritten, by swapping columns, as
\[
\sum_n \sum_{\ell=1}^{m-1} (a_n^\ell)_t \left|\gamma_{n+1}, \dots, Y_{n+\ell}, \dots, \gamma_{n+\ell}\right| = (-1)^{m-1}\sum_n \sum_{\ell=1}^{m-1} (a_n^\ell)_t \left|Y_{n+\ell}, \gamma_{n+1}, \dots,  \gamma_{n+m-1}\right|.
\]

Up to a sign, this expression is the dot product of the first row of $Q_n^Y$ with its first entry removed with the vector $(a_n^\ell)_t$. Now assume that $\a$ evolves by a Hamiltonian flow with Hamiltonian $f$ with respect to the reduced bracket, i.e. $(a_n^\ell)_t= X^f=\Pc_n \delta_n f$. When $Y = X^g$, the first row of $Q_n^g$ with its first entry removed is equal to $(-1)^{m} \delta_n g$, giving  $\sum_n \sum_{\ell=1}^{m-1} (a_n^\ell)_t \left|\gamma_{n+1}, \dots, Y_{n+\ell}, \dots, \gamma_{n+\ell}\right|=\{f, g\}_R$. The remaining terms produce another copy of the same, which takes care of the factor of 1/2 in the definition of  $\omega_1$.
\end{proof}

We are now ready for one of the main results of this work.
\begin{theorem}\label{compatibility}The bracket 
\[
\{f, g\} = \{f,g\}_R + \{f,g\}_0
\]
is a Poisson bracket. Therefore, the reduced Poisson brackets $\{\, \ ,\ \}_R$ and $\{\ \, ,\ \}_0$ are compatible.
\end{theorem} 
\begin{proof}
We only need to verify the Jacobi identity. From~\eqref{diffsymp},
\[
X^f(g):=\rd(g\circ\pi)(X^f)=\omega_1(X^f,X^g)=\{f,g\}_R\circ\pi,
\]

and from Theorem~\ref{omega2br},
\[
 \omega_2(X^f, X^g) = \{f, g\}_0\circ\pi.
 \]
Denoting cyclical summations with $\displaystyle \sum_\rcirclearrowleft$:
\[
\begin{split}
& \sum_\rcirclearrowleft \{\{f,g\}_R+\{f,g\}_0,h\}_R+\{\{f,g\}_R+\{f,g\}_0,h\}_0 =  \sum_\rcirclearrowleft \{\{f,g\}_0,h\}_R+\{\{f,g\}_R,h\}_0 \\
&=\sum_\rcirclearrowleft \{\{f,g\}_0,h\}_R\circ\pi+\{\{f,g\}_R,h\}_0\circ\pi = -\sum_\rcirclearrowleft X^h(\omega_2(X^f,X^g)) + \omega_2(X^{\{f,g\}_R}, X^h).
\end{split}
 \]
 Now, recall that
 \[
 \begin{split}
[X^f,X^g](h\circ\pi)&  = X^f(X^g(h\circ\pi)) - X^g(X^f(h\circ\pi)) = X^f(\{g,h\}_R\circ\pi))- X^g(\{f,h\}_R\circ\pi))\\
& =\{f, \{g,h\}_R\}_R\circ\pi-\{g, \{f,h\}_R\}_R\circ\pi 
= -\{h, \{f,g\}_R\}_R\circ\pi = X^{\{f,g\}_R}(h\circ\pi).
\end{split}
\]
Therefore,
\[
\begin{split}
& \sum_\rcirclearrowleft \{\{f,g\}_R+\{f,g\}_0,h\}_R+\{ \{f,g\}_R+\{f,g\}_0,h\}_0,h\}_0 \\
& =  -\big(\sum_\rcirclearrowleft X^h(\omega_2(X^f,X^g)) - \omega_2([X^f,X^g], X^h)\big)=\rd\omega_2(X^f, X^g, X^h)=0,
\end{split}
\]
since $\omega_2 $ is a closed form.
\end{proof}
While Theorem~\ref{compatibility} demonstrates the compatibility of the reduced brackets in a simple and direct way, the implications on integrability are not yet clear since $\omega_1$ and $\omega_2$ have non-trivial kernels. 
We next investigate whether we have a bi-Hamiltonian pair.  We begin by proving that $\omega_1$ is also a closed 2-form;  the proof relies on a general expression for $\omega_1$ that is valid for arbitrary (not necessarily arc length-preserving) vector fields $X, Y \in \V_N$.\smallskip
\begin{proposition}\label{omega1alt}
Assume $X, Y \in \V_N$ and let $Q^X$ and $Q^Y$ define the associated frame evolutions as in~\eqref{frameev}. Then,
\begin{equation}\label{oformula}
\begin{split}
 \omega_1(X,Y) &= -\frac{1}{2}\sum_{n=1}^N \left\{ [Q_n^Y, Q_n^X]_{m,m}+\left(K_n^{-1}[Q_n^Y, Q_n^X]K_n\right)_{m,m}\right.\\
    & -\mathrm{tr}(Q_n^X)\big((Q_n^Y)_{m,m}  
    - (K_n^{-1}Q_n^YK_n)_{m,m}\big)\\
    & +\mathrm{tr}(Q_n^Y)\big((Q_n^X)_{m,m}-(K_n^{-1}Q_n^XK_n)_{m,m}\big)\Big\},
 \end{split}
\end{equation}
 where  $(M)_{i,j}$ denotes the $(i,j)$-th entry of matrix $M$.
\end{proposition}
\begin{remark}
The second and third lines of~\eqref{oformula} vanish whenever $X$ and $Y$ are arc length-preserving vector fields.
\end{remark}
\begin{proof}
By replacing $\gam_{n+r-1}=\rho_n\e_r$, $X_{n+r-1}=\rho_n Q^X_n \e_r$, and $Y_{n+r-1}=\rho_n Q_n^Y \e_r$, we compute:
\[
  \begin{split}
&  |\gam_n,\dots,\stackrel{r}{\overbrace{X_{n+r-1}}},\dots,Y_{n+m-1}|- |\gam_n,\dots,\stackrel{r}{\overbrace{Y_{n+r-1}}},\dots,X_{n+m-1}| \\
&  =   |\e_1, \ldots , \stackrel{r}{\overbrace{Q^X_n \e_r}}, \ldots, Q^Y_n \e_m| - |\e_1, \ldots,  \stackrel{r}{\overbrace{Q^Y_n \e_r}}, \ldots, Q^X_n \e_m |\\
&  =  (Q^X_n)_{r,r} (Q^Y_n)_{m,m}|\e_1, \ldots,  \stackrel{r}{\overbrace{\e_r}}, \ldots, \e_m| + (Q^X_n)_{m,r} (Q^Y_n)_{r,m}|\e_1, \ldots,  \stackrel{r}{\overbrace{\e_m,}} \ldots, \e_r| \\ 
&  -(\text{terms with X replaced by Y})\\
 & = (Q^X_n)_{r,r} (Q^Y_n)_{m,m} - (Q^X_n)_{m,r} (Q^Y_n)_{r,m} -(\text{terms with X replaced by Y}). \end{split}
 \]
This gives
 \begin{equation}\label{uno}
 \begin{split}
 \sum_{r=1}^{m-1} & \Big(|\gam_n,\dots,\stackrel{r}{\overbrace{X_{n+r-1}}},\dots,Y_{n+m-1}|- |\gam_n,\dots,\stackrel{r}{\overbrace{Y_{n+r-1}}},\dots,X_{n+m-1}|\Big) \\
 & = \mathrm{tr}(Q_n^X)(Q_n^Y)_{m,m}-\mathrm{tr}(Q_n^Y)(Q_n^X)_{m,m}+[Q_n^Y, Q_n^X]_{m,m}.
 \end{split}
 \end{equation}
A slightly longer, but straightforward computation, making use of the relations: 
 \[
 \sum_{\ell=1}^{m-1} a_n^\ell Q_n^Y\e_{\ell+1} + (-1)^{m-1}Q_n^Y\e_1 = Q_n^Y K_n \e_m, \quad\quad (-1)^{m-1} e_1^TQ_n^Y = \e_m^T K_n^{-1}Q_n^Y,
 \]
yields
\begin{equation}\label{dos}
\begin{split}
& \sum_{r=1}^{m-1} \bigg\{(-1)^{m-1}\sum_{\ell=1}^{m-1}a_n^{\ell}\Big(|Y_{n+\ell},\dots,\stackrel{r+1}{\overbrace{X_{n+r}}}, \dots, \gam_{n+m-1}|- |X_{n+\ell},\dots,\stackrel{r+1}{\overbrace{Y_{n+r}}}, \dots, \gam_{n+m-1}|\Big) \\
& + |Y_{n},\dots,\stackrel{r+1}{\overbrace{X_{n+r}}}, \dots, \gam_{n+m-1})- |X_{n},\dots,\stackrel{r+1}{\overbrace{Y_{n+r}}}, \dots, \gam_{n+m-1}|\bigg\}\\
&
 = \Big(K_n^{-1}[Q_n^Y, Q_n^X]K_n\Big)_{m,m}+\mathrm{tr}(Q_n^X)\Big(K_n^{-1}Q_n^YK_n\Big)_{m,m}-\mathrm{tr}(Q_n^Y)\Big(K_n^{-1}Q_n^XK_n\Big)_{m,m}.
 \end{split}
\end{equation}
Formula~\eqref{dos} is valid for arbitrary $(1,m)$-entry of $K_n$, i.e.  when $|\gamma_n, \dots, \gamma_{n+m-1}|$ is not a constant. Replacing \eqref{uno} and \eqref{dos} in the expression~\eqref{omega1} for $\omega_1$ gives formula~\eqref{oformula}.

\end{proof}

The 2-form $\omega_2$ also has an analogous representation as described by the following 
\begin{proposition}\label{omega2alt}
Assume $X, Y \in \V_N$ and let $Q^X$ and $Q^Y$ be the matrices define the associated frame evolutions as in~\eqref{frameev}. Then,
\begin{equation}\label{oformula2}
\omega_2(X,Y) = \sum_{n=1}^N \left[Q_n^Y, Q_n^X\right]_{1,1}.
\end{equation}
\end{proposition}
The proof is similar to the proof of Proposition~\ref{omega1alt}.
\begin{remark}

It is a natural to wonder whether any of the $2$-forms 
\[
\omega_{r+s}(X,Y) =  \sum_{n=1}^N \left[Q_n^Y, Q_n^X\right]_{r,s},\ \ s\geq r,
\]
which can be easily shown to be closed, induce a Poisson bracket on the moduli space $G^N/H^N$. In particular for $s=r$, in which case the 2-form $\omega_{2r}$ is also exact:
 \[
 \omega_{2r}=\rd\theta_{2r}, \qquad \qquad \theta_{2r}(X) = \frac{1}{2}\sum_{n=1}^N |\gamma_n, \gamma_{n+1}, \dots, \stackrel{r+1}{\overbrace{X_{n+r}}}, \dots, \gamma_{n+m-1}|.
 \]
However, in dimension $m=3$,  a quick calculation shows that $\omega_{r+s}$ defines a Poisson bracket only when $r=s=1$. We conjecture that $\omega_1$ and $\omega_2$ are the only two Poisson brackets that can be generated in terms of entries of the commutators $\left[Q_n^Y, Q_n^X\right]$.
\end{remark}

We now proceed with computing the differential of $\omega_1$.

Given an invariant vector field $X\in \V_N$, with $(X|_\gam)_n=\sum_{\ell=0}^{m-1} r_n^\ell \gam_{n+\ell}$ we introduce the reparametrization operator $\mathscr{P}: \V_N \to \V_N^1$:
\begin{equation}
\label{repar}
(\mathscr{P}X)_n:=\hat r^0_n\gam_n + r^1_n\gam_{n+1}+ \ldots r^{m-1}_n \gam_{n+m-1},
\end{equation}
where the coefficients $\hat{r}^0_n$'s are uniquely determined by the arc length preservation condition~\eqref{alpreserve}. Since $(\mathscr{P} X)_{n+r}-X_{n+r}=(\hat r^0_{n+r} -r^0_{n+r})\gam_{n+r}=(\hat r^0_{n+r} -r^0_{n+r})\rho_n\e_{r+1}$, and $Q_n^X=\rho_n^{-1}\begin{pmatrix} X_n & X_{n+1} \ldots & X_{n+m-1} \end{pmatrix}$, then
 \begin{equation}\label{Qeq}
 Q_n^{\mathscr{P}X} = Q_n^X + \rd_n(X),
 \end{equation}
  where $\rd_n(X)$ is the diagonal matrix
 \begin{equation}\label{d}
 \rd_n(X) = \mathrm{diag}(\hat{r}^0_n-r^0_ n, \hat{r}^0_{n+1}-r^0_ {n+1}, \dots, \hat{r}^0_{n+m-1}-r^0_ {n+m-1}). 
 \end{equation}
 Since $\mathscr{P} X$ is arc lengh-preserving, computing the trace of both sides of~\eqref{Qeq}, we obtain
 \begin{equation}\label{xtr}
0= \mathrm{tr}(Q_n^X)+\big(1 + \T + \ldots \T^{m-1}\big) (\hat{r}^0_n-r^0_ n),
 \end{equation}
which can be solved uniquely for $\hat{r}^0_n$, $n=1, \ldots, N,$ since $N$ and $m$ are assumed to be co-prime.

 \begin{remark}
 From now on, we will use the notation $\X:=\mathscr{P} X$ to denote the reparametrized vector field.
 \end{remark}


  \begin{lemma}\label{extdo1Lem}
Assume $\X,\Y,\Z\in \V^1_N$ are commuting vector fields and let  $\gam\in\M_N^1$. Then, 
\begin{equation}
\label{extdo1}
\rd\omega_1|_\gamma(\X,\Y,\Z) = \sum_\circlearrowright \Z(\gam)\omega_1(\X,\Y)|_\gam.
\end{equation}
\end{lemma} 

 \begin{proof} 
Let $X, Y, Z$ be commuting vector fields in $\V_N$, not necessarily arc length-preserving. Substituting expression \eqref{Qeq} into formula \eqref{oformula} and using $[\rd_n(X), \rd_n(Y)]=0$, we compute the exterior derivative of $\omega_1$ as follows:
\begin{equation}\label{lastJacobi}
\begin{split}
\rd\omega_1(X,Y,Z) & = \sum_\circlearrowright Z\omega_1(\X, \Y) \\
&   - \frac{1}{2} \sum_\circlearrowright  \sum_n Z_n\bigg\{ -[Q_n^\Y, \rd_n(X)]_{m,m}-[\rd_n(Y), Q_n^\X]_{m,m}-\left(K_n^{-1}[Q_n^\Y, \rd_n(X)]K_n\right)_{m,m} \\ 
&\, \ \ -\left(K_n^{-1}[\rd_n(Y), Q_n^\X]K_n\right)_{m,m}
 +\mathrm{tr}(\rd_n(X))\Big[(Q_n^\Y)_{m,m}-\left(K_n^{-1}Q_n^\Y K_n\right)_{m,m}\Big] \\
& -\mathrm{tr}(\rd_n(Y))\Big[(Q_n^\X)_{m,m}-\left(K_n^{-1}Q_n^\X K_n\right)_{m,m}\Big] \bigg\}.\\
\end{split}
\end{equation}
First, note that $[Q_n^\Y, \rd_n(X)]_{m,m}=[\rd_n(Y), Q_n^\X]_{m,m}=0$, since $\rd_n(X)$, $\rd_n(Y)$ are diagonal matrices. 
Also, when $X=\X$ is an arc length-preserving vector field,~\eqref{Qeq} implies $\rd_n(\X)=0$; it follows that the only non-vanishing terms left in expressions such as $Z_n\left(K_n^{-1}[Q_n^\Y, \rd_n(X)]K_n\right)_{m,m}$ when evaluated at $X=\X$, will be $(K_n^{-1}[Q_n^\Y, Z_n(\rd_n(X))]K_n)_{m,m}.$

Writing 
\[
\frac{\rd \gam_n}{\rd t_1} = X_n,\quad \frac{\rd \gam_n}{\rd t_2} =  Y_n, \quad \frac{\rd \gam_n}{\rd t_3}  = Z_n,
\]
and taking the traces of the relations 
\[
\frac{\rd }{\rd t_3} Q_n^X = \frac{\rd}{\rd t_1} Q_n^Z + [Q_n^Z, Q_n^X],\qquad  \frac{\rd }{\rd t_3} Q_n^Y = \frac{\rd}{\rd t_2} Q_n^Z + [Q_n^Z, Q_n^Y],
\]
  we compute 
\[
\frac{\rd }{\rd t_3} \mathrm{tr}(Q_n^X) = \frac{\rd}{\rd t_1} \mathrm{tr}(Q_n^Z),\qquad  \frac{\rd }{\rd t_3} \mathrm{tr}(Q_n^Y)= \frac{\rd}{\rd t_2} \mathrm{tr}(Q^Z_n).
\]
Since~\eqref{Qeq} implies $\mathrm{tr}(Q_n^X)=-\mathrm{tr}(\rd_n(X))$ and similar identities when $X$ is replaced with $Y$ and with $Z$, then
\[
\frac{\rd }{\rd t_3} \mathrm{tr}(\rd_n(X)) =\frac{\rd}{\rd t_1} \mathrm{tr}(\rd_n(Z)),\qquad  \frac{\rd }{\rd t_3} \mathrm{tr}(\rd_n(Y))= \frac{\rd}{\rd t_2} \mathrm{tr}(\rd_n(Z)).
\]
Again from~\eqref{Qeq} and the invertibility of $1 + \T + \ldots \T^{m-1}$, it also follows that
\[
\frac{\rd }{\rd t_3} \rd_n(X) =\frac{\rd}{\rd t_1}\rd_n(Z),\qquad  \frac{\rd }{\rd t_3} \rd_n(Y)= \frac{\rd}{\rd t_2}\rd_n(Z).
\]

Using these relations it is straightforward to check that all remaining terms in the second cyclic sum vanish, giving
\[
\rd\omega_1(X,Y,Z) = \sum_\circlearrowright Z\omega_1(\X, \Y). 
\]
Evaluating both sides at $\X, \Y, \Z \in \V^1_N$ concludes the proof.
 \end{proof}
%
%
%
%
%

\begin{theorem}\label{o1closed}The $2$-form $\omega_1$ is a closed form on $\M_N^1$. 
 \end{theorem}
 
Before providing an abstract proof of this result, we derive a simplified formula for $\rd \omega_1$ and use it in a computational proof for dimension 2.

For $X, Y, Z$ arc lenght-preserving vector fields, formula~\eqref{extdo1} reduces to        
 \[
\rd \omega_1(X, Y, Z)=  - \frac{1}{2}  \sum_{n=1}^N \sum_\circlearrowright Z_n\Big( [Q_n^Y, Q_n^X]_{m,m}+\left(K_n^{-1}[Q_n^Y, Q_n^X]K_n\right)_{m,m}\Big).
 \]
The identity $[Q_n^Y, Q_n^X]=X(Q_n^Y)-Y(Q_n^X)$ implies that $\displaystyle \sum_\circlearrowright Z_n([Q_n^Y, Q_n^X])=0$. Moreover, one can check that the $m$-th row of $K_n^{-1}Z(K_n)$ is the zero vector. It follows that
 \[
 \rd \omega_1(X, Y, Z)=  - \frac{1}{2}  \sum_{n=1}^N \sum_\circlearrowright \left(K_n^{-1}[Q_n^Y, Q_n^X]Z(K_n)\right)_{m,m}.
 \]
 Finally, using the identities:
 \[
\e_m^T K_n^{-1} =(-1)^{m-1} \e_1, \qquad Z(K_n)\e_m=\begin{pmatrix} 0& 0\\ 0&Z(\a_n)\end{pmatrix}\e_m=\sum_{j=2}^m Z(a_n^{j-1})\e_j,
\]
we get the simplified formula
\begin{equation}
\label{do1simple}
 \rd \omega_1(X, Y, Z)=  - \frac{1}{2}  \sum_{n=1}^N \sum_\circlearrowright \sum_{j=2}^m Z(a_n^{j-1})[Q_n^Y, Q_n^X]_{1,j}.
 \end{equation}
  \begin{example}{\em (2-dimensional case)}
From  $K_n = \begin{pmatrix} 0&-1\\ 1&a_n\end{pmatrix}$ and 
the condition
\[
K_n^{-1} \frac{\rd}{\rd t} K_n= \begin{pmatrix} 0& \rd a_n/\rd t\\ 0&0\end{pmatrix} = Q_{n+1} - K_n^{-1}Q_n K_n,
\]
we obtain the following general expression
\[
Q_n = \begin{pmatrix} \Rc^{-1} a_n q_n & q_n\\ -q_{n-1}&-\Rc^{-1} a_nq_n\end{pmatrix}.
\]
For some $q_n$.  (Note that $\Rc = 1+\T$ is invertible for odd period $N$.) Writing $Q_n^X = Q_n(q)$, $Q_n^Y = Q_n(p)$, $Q_n^Z = Q_n(r)$, we compute
\[
[Q_n^X, Q_n^Y] = \begin{pmatrix} p_nq_{n-1} - q_np_{n-1} & 2(p_n\Rc^{-1}a_nq_n - q_n \Rc^{-1}a_np_n)\\ 2(p_{n-1}\Rc^{-1}a_nq_n - q_{n-1} \Rc^{-1}a_np_n)& p_{n-1}q_n-q_{n-1}p_n
\end{pmatrix},
\]
and
\[
Q_{n+1}^Z - K_n^{-1}Q_n^Z K_n=\begin{pmatrix} 0& r_{n+1}-r_{n-1}+a_n \Rc^{-1}(\T-1)a_n r_n\\ 0&0 \end{pmatrix},
\]
giving
\[
Z(a_n)=r_{n+1}-r_{n-1}+a_n \Rc^{-1}(\T-1)a_n r_n.
\]
Then, equation~\eqref{do1simple} with $a_n^1=a_n$ becomes
\[
\begin{split}
\rd \omega_1(X, Y)= & - \frac{1}{2}  \sum_{n=1}^N \sum_\circlearrowright ([Q_n^Y, Q_n^X])_{1,2}Z(a_n) \\
= & \sum_{n=1}^N \sum_\circlearrowright (q_n\Rc^{-1}a_np_n - p_n \Rc^{-1}a_nq_n) (r_{n+1}-r_{n-1}+a_n \Rc^{-1}(\T-1)a_n r_n).
\end{split}
\]
Using the identity $\Rc^{-1} (\T-1)=1 -2\Rc^{-1}$, it is easy to check that
\[
\begin{split}
\sum_\circlearrowright & (q_n\Rc^{-1}a_np_n - p_n \Rc^{-1}a_nq_n) a_n \Rc^{-1}(\T-1)a_n r_n\\
&=\sum_\circlearrowright \left[a_n^2r_n(q_n\Rc^{-1}a_np_n - p_n \Rc^{-1}a_nq_n)-2a_n(\Rc^{-1}a_nr_n)(q_n \Rc^{-1}a_np_n-p_n\Rc^{-1}a_nq_n)\right]=0.
\end{split}
\]
For the remaining cyclical sum $\displaystyle \sum_{n=1}^N\sum_\circlearrowright (q_n\Rc^{-1}a_np_n - p_n \Rc^{-1}a_nq_n) (r_{n+1}-r_{n-1})$, we group its terms as follows:
\[
\begin{split}
& \sum_{n=1}^N \left( r_{n+1} q_n R^{-1} a_n p_n - r_{n+1} p_n R^{-1} a_n q_n -p_{n-1} r_n R^{-1} a_n q_n +q_{n-1} r_n R^{-1} a_n p_n \right)\\
& = \sum_{n=1}^N \left( r_{n+1} q_n R^{-1} a_n p_n - r_{n+1} p_n R^{-1} a_n q_n -p_{n} r_{n+1} R^{-1} \T a_n q_n +q_{n} r_{n+1} R^{-1} \T a_n p_n \right)\\
&=\sum_{n=1}^N \left(  r_{n+1} q_n R^{-1} (1+\T) a_n p_n - r_{n+1} p_n R^{-1}(1+\T) a_n q_n \right) = \sum_{n =1}^N \left(  r_{n+1} q_n a_n p_n - r_{n+1} q_na_n p_n \right)= 0,
\end{split}
\]
using  $1+\T = \Rc$. The remaining terms are dealt with in a similar way.
\end{example}

 \begin{proof}[Proof of Theorem~\ref{o1closed}]
Let $\gam \in \M_N^1$ and assume for simplicity that  $X, Y, Z$ are commuting vector fields in $\V_N^1$, i.e. $X=\X, Y=\Y, Z=\Z$.
From Lemmae~\ref{extdo1Lem}, 
 \[
 \rd\omega_1(X,Y,Z)(\gamma) =\sum_\circlearrowright Z(\gamma)\omega_1(X,Y)(\gamma).
 \]
 We now recall Theorem \ref{redth} and its implications. The Poisson bracket $\{\, \ ,\ \}_R$ is a reduction of Semenov-Tian-Shansky's twisted bracket \eqref{twisted}. Since its symplectic leaves (at the {\emph{invariant level}})
  are the $G^N$-orbits under the gauge action
 \[
h_n \to  g_{n+1}h_n g_{n}^{-1}.
 \]
it follows that the tangent space to a symplectic leaf ({\it before reduction}) comprises the Lie algebra invariant elements of the form
 \[
P_{n+1} h_n -h_nP_{n}
 \]
$h_n\in G$, $P_n \in \g$. A complement to the leaves is then obtained by solving  $P_{n+1} = h_n^{-1} P_n h_n, \forall n$. When the twisted bracket restricted to $\gamma_n$ and reduced, its symplectic leaves are described in terms evolutions of Maurer-Cartan matrices $K_n$, and the complement is given by solutions of $Q_{n+1} = \rho_n^{-1} Q_n \rho_n$, for which $\rd K_n/\rd t= 0$.

 
Next, split $\V_N^1|_\gamma$ into vector fields in the kernel of $\omega_1$ (i.e. describing directions in $\V_N^1|_\gamma$ for which $\rd K_n/\rd t$ vanishes), and vector fields whose push forwards on the  invariant manifold are tangent to the symplectic leaves of  $\{\, \ ,\ \}_R$. We can assume that the latter are also Hamiltonian vector fields on the space of geometric invariants of $\gamma$, since they generate the tangent space to the symplectic leaf.  We can now consider three scenarios.

If $X, Y$ and $Z$ are all tangent to a  given symplectic leaf, then, since $\omega_1$ is the pull-back of the symplectic structure defined on the corresponding symplectic leaf, it must be a closed form. Hence 
 \[
 \rd\omega_1(X,Y,Z)(\gamma) = 0.
 \]
 If two or three of $X,Y,Z$ are in the kernel of $\omega_1$, then we immediately have $\rd\omega_1(X,Y,Z)(\gamma) = 0$. The remaining case is when only one vector field, say $Z$, belongs to the kernel of $\omega_1$. The only term we need to compute is 
 $Z(\gamma)\left(\omega_1(X,Y)\right)$,
 since the other terms will vanish.  By assumption,  $\rd \gamma_n/\rd t = Z(\gamma)_n$ induces the stationary flow
 $\rd K_n/\rd t = Q_{n+1} - K_n^{-1} Q_n K_n = 0$ on the geometric invariants $a_n^i$'s. Since  $\omega_1(X,Y)$ is an invariant function of the  $a_n^i$'s, then  $Z(\gamma)\left(\omega_1(X,Y)\right) = 0$.  \end{proof}
 
 \begin{remark}  If $\gamma$ is not arc length parametrized, $\rd\omega_1|_\gamma$ in general will not vanish. We will discuss this point in the next section.
 \end{remark}
 
 This proof provides a good understanding of the role of $\omega_1$.  Elements of its kernel push forward to  transverse sections to the symplectic leaves of $\{\, \ ,\ \}_R$ under the map $X \to Q^X$, where $Q^X$ is the matrix defining the associated frame evolution.  Whereas, the restriction of $\omega_1$ to tangent vector fields that push forward to the symplectic leaf, coincides with the natural symplectic structure on the orbit.

\section{The kernels of the pre-symplectic forms}\label{kernelsection}

Shifting the point of view from the AGD flows on geometric invariants to the associated polygonal evolutions presents several advantages. On the one hand, not all vector fields on the Poisson space of geometric invariants will admit a lift to invariant---monodromy preserving---geometric flows (i.e.~an invariant polygonal evolution). On the other hand, given a geometric flow , then the induced flow on the space of geometric invariants is automatically tangent to the leaves of the Poisson structure.
In addition, as both $\omega_1$ and $\omega_2$ are trivially reducible to the moduli space, they will induce Poisson structures on the submanifolds of fixed monodromy conjugacy class, $[T] = const$.

In this final section we describe the kernels of $\omega_1$ and $\omega_2$, show that  $\omega_1$ is a symplectic form on submanifolds where $[T]=const$, and discuss how we may tackle the  construction of an integrable hierarchy of vector fields for polygonal flows.

\subsection{The kernel of $\omega_1$}
\begin{theorem}\label{kernel} Assume $\gamma=\{\gamma_n\}\in \M_n^1$ is a twisted polygon with monodromy $T\in \SL(m,\RR)$. Then, the dimension of the kernel of $\omega_1|_\gamma$  equals the dimension of the isotropy subgroup of $T$ (under the adjoint action).  

Furthermore, if $\mathfrak{a}$ belongs to the isotropy subalgebra, and we define $Q_1 = \mathfrak{a}$, $Q_{n+1} = K_n^{-1}Q_{n}K_n$, then the vector field $X = \{X_n\}\in\V_n^1$, with components $X_n = Q_n \e_1$ belongs to the kernel of $\omega_1$.
\end{theorem}
 \begin{proof}
 Recall, from Theorem~\ref{Xf} and Proposition~\ref{Qtovariations}, that the matrix $Q_n^f$ associated to a Hamiltonian functional $f$ satifies $(Q_n^f)_{-1} = (\nabla'_n \F)_{-1}$, where $\F$ is an invariant extension of $f$. Then, letting $f, h$ be  smooth functionals on $G^N/H^N$ and $X^f, X^h$ be the unique vector fields in $\V^1_N$ given by Theorem~\ref{Xf}, 
 \[
 \omega_1(X^f, X^h) = \sum_{n=1}^N\langle (\nabla_n \F)_{1}-\T(\nabla'_n\F)_1, (Q_n^h)_{-1}\rangle.
 \]
 Since condition~\eqref{maincond}, relating the right and left gradients, and condition~\eqref{determineQ}, determining $Q$, are algebraic conditions, and since $\nabla'_n\F$ depends solely on $\delta f$ via~\eqref{onecomponent}, one can replace $\delta f$ in~\eqref{onecomponent} with an arbitrary vector and proceed to determine the remaining entries of the matrix using~\eqref{determineQ}. The resulting matrix will not in general correspond to a gradient.
 Proceeding as described, for $X, Y\in \V_N^1$ and given $Q^X, Q^Y$ the associated matrices, we let $S_n^X$ be the solution of
 \begin{equation}\label{Sn}
 (S_n^X)_{-1} = (Q_n^X)_{-1}; \qquad\qquad S^X_{n+1}-K_n^{-1}S^X_nK_n\in \g_1.
 \end{equation}
(Note that, if $X = X^f$, then $S^X = \nabla'\F$.)
Then, 
 \[
 \omega_1(X,Y) = \sum_{n=1}^N\langle \left(K_n^{-1}S_n^XK_n-S_{n+1}^X\right)_1, (Q_n^Y)_{-1}\rangle.
 \]
 Assume $X$ is an element of the kernel of $\omega_1$, i.e.  $\omega_1(X, Y) = 0, \forall Y\in \V_N^1$. Since an arbitrary $Y\in \V_N^1$ arbitrary implies an arbitrary $Q^Y$, then $X$ belongs to the kernel of $\omega_1$ if and only if $
 \left(K_n^{-1}S_n^XK_n-S_{n+1}^X\right)_1 = 0, \ \forall n.$
 Since $S^X_{n+1}-K_n^{-1}S^X_nK_n\in \g_1$, we obtain the condition
 \[
 S^X_{n+1}=K_n^{-1}S^X_nK_n, \quad \forall n,
 \]
giving
 \[
 S^X_{n+N} = K_{n+N-1}^{-1}\dots K_n^{-1} S_n^X K_n\dots K_{n+N-1} =\left(\rho_n^{-1}T^{-1}\rho_n \right) S_n^X \left(\rho_n^{-1}T\rho_n\right).
 \]
Since $\gamma_n$ is a twisted polygon of period $N$ and monodromy $T$, then $S^X_{n+N}=S^X_n$, so $[\rho_n^{-1} T\rho_n, S^X_n]=0$. This proves the first  claim.

For the second claim,  we only need to check that $[\rho_1^{-1} T\rho_1, S^X_1]=0$, since $S^X=Q^X$, which is the case because $S^X_1=\mathfrak{a}$, an element of the isotropy subalgebra.
 \end{proof}
 
\begin{corollary}\label{rigidmotion}
The kernel of $\omega_1$ is generated by the infinitesimal rigid motions that leave the monodromy invariant.
\end{corollary}
\begin{proof}
Given  an element of the isotropy subalgebra $\mathfrak{a}\in \mathfrak{sl}(m, \RR)$, denote by $g_{\mathfrak{a}}(t):=\e^{t\mathfrak{a}}$ the associated one-parameter subgroup of $SL(m, \RR)$. For a given (twisted) polygon $\{\gamma_n\}$, the corresponding variation is $\{ g_{\mathfrak{a}}(t)\gamma_n \}$. The $n$-th component of  the variation vector field is computed as follows:
\begin{equation}
\label{rigidVF}
\dot{\gamma}_n=\left. \ddt g_{\mathfrak{a}}(t)  \gamma_n \right|_{t=0}= \mathfrak{a} \gamma_n.
\end{equation}
Note that the vector field  \eqref{rigidVF} is, by construction, locally arclength preserving. This can be also checked directly by differentiating $|g_\mathfrak{a}(t)(\gamma_n,\dots,\gamma_{n+m-1})| = 1$ and computing $| \mathfrak{a}  \gamma_n,  \ldots, \gamma_{n+m-1}|+
| \gamma_n, \ldots, \mathfrak{a} \gamma_{n+m-1}|=0$.

From \eqref{rigidVF}, we have
\[
\dot{\rho}_n=\mathfrak{a} \rho_n=\rho_n \rho_n^{-1} \mathfrak{a} \rho_n=\rho_n Q_n, \qquad Q_n:=\rho_n^{-1} \mathfrak{a} \rho_n.
\]
Moreover, from $\rho_n Q_n=\mathfrak{a} \rho_n$, compute $\rho_{n+1}Q_{n+1}=\mathfrak{a} \rho_{n+1} \iff \rho_n K_n Q_{n+1}=\mathfrak{a} \rho_n K_n$, giving
\[
Q_{n+1}=K_n^{-1} Q_n K_n.
\]
The choice $\rho_1=\mathrm{Id}$ (identity matrix) as initial condition, gives $Q_1=\mathfrak{a}$. It follows from Theorem \ref{kernel}, that \eqref{rigidVF} is in the kernel of $\omega_1$, conversely, every element in the kernel is an infinitesimal rigid motion as defined by \eqref{rigidVF}.
\end{proof}
\noindent

{\begin{lemma}\label{rigidhamil}
The Hamiltonian functional associated to an infinitesimal rigid motion by means of $\omega_2$ is given by
\begin{equation}
\label{charges2}
h_\mathfrak{a}(\gamma)=\sum_{n=1}^N \left|\mathfrak{a}\gam_n, \gam_{n+1}, \ldots, \gam_{n+m-1}\right|,   
\end{equation}
where $\mathfrak{a} \in \mathfrak{sl}(m, \RR)$ belongs the isotropy subalgebra of the monodromy.
\end{lemma}
\begin{proof}
Given $g_{\mathfrak{a}}(t):=\e^{t\mathfrak{a}} \in \SL(m,  \RR)$, the one-parameter subgroup generated by $\mathfrak{a}$, the following holds for every $t$:
\begin{equation}
\label{invardet}
|g_{\mathfrak{a}}(t)v_n, g_{\mathfrak{a}}(t)v_{n+1}, \ldots, g_{\mathfrak{a}}(t)v_{n+m-1}|=|v_n, v_{n+1}, \ldots, v_{n+m-1}|, 
\end{equation}
and for any $m$-tuple of vectors $v_j\in \RR^m$.
As before, differentiating both sides of \eqref{invardet} with respect to $t$ and evaluating at $t=0$, one gets
\begin{equation}
\label{infinsym}
|\mathfrak{a}v_n, v_{n+1}, \ldots, v_{n+m-1}| + \sum_{r=1}^{m-1} |v_n, v_{n+1}, \ldots, \stackrel{r+1}{\overbrace{\mathfrak{a}v_{n+r}}},\ldots, v_{n+m-1}|=0.
\end{equation}
Let now $Y=\{Y_n\}\in \V_N^1$, and let $X_\mathfrak{a}=\{\mathfrak{a} \gam_n\}$ be the infinitesimal rigid symmetry vector field generated by $\mathfrak{a}$. Then, using \eqref{infinsym},
\[
\begin{split}
\omega_2(Y, X_\mathfrak{a}) & =-\frac{1}{2}\sum_{n=1}^N \Big( \sum_{r=1}^{m-1} |Y_n, \gam_{n+1}, \ldots, \stackrel{r+1}{\overbrace{\mathfrak{a}\gam_{n+r}}},\ldots, \gam_{n+m-1}|  \\ & - \sum_{r=1}^{m-1} | \mathfrak{a}\gam_n,  \gam_{n+1}, \ldots, \stackrel{r+1}{\overbrace{Y_{n+r}}},\ldots, \gam_{n+m-1}|\Big)\\
&=\frac{1}{2} \sum_{n=1}^N \Big( |\mathfrak{a}Y_n, \gam_{n+1}, \ldots , \gam_{n+m-1}| +\sum_{r=1}^{m-1} | \mathfrak{a}\gam_n,  \gam_{n+1}, \ldots, \stackrel{r+1}{\overbrace{Y_{n+r}}},\ldots, \gam_{n+m-1}|\Big)\\
& =\rd h_\mathfrak{a}[Y],
\end{split}
\]
where $\rd h_\mathfrak{a}[Y]$ denotes the differential of $h_\mathfrak{a}$ at $Y$. 
\end{proof}
\noindent

\begin{corollary}\label{nointegrable}
Let $\gam\in\M_N^1$. Given the symmetry vector field $X_\mathfrak{a}=\{\mathfrak{a} \gam_n\}$, there is no  vector field $W_\mathfrak{a}\in \V_N^1$ satisfying
\[
\omega_1(W_\mathfrak{a}, Y)|_\gam =\omega_2(X_\mathfrak{a}, Y)|_\gam, \quad \forall Y\in \V_N^1.
\]
(That is, there are no ``negative" flows generated by the symmetry vector fields.)
\end{corollary}
\begin{proof}
If such $W_\mathfrak{a}$ existed, then 
\[
0=\omega_1(W_\mathfrak{a}, X_\mathfrak{b}) =\omega_2(X_\mathfrak{a}, X_\mathfrak{b}).
\]
where $\mathfrak{b}$ is a different element of the isotropy subalgebra. However,
the right-hand-side is never zero. In fact,
\[
\begin{split}
\omega_2(X_\mathfrak{a}, X_\mathfrak{b})=&\frac{1}{2} \sum_{n=1}^N \Big( |\mathfrak{a}\mathfrak{b}\gam_n, \gam_{n+1}, \ldots , \gam_{n+m-1}| +\sum_{r=1}^{m-1} | \mathfrak{a}\gam_n,  \gam_{n+1}, \ldots, \stackrel{r+1}{\overbrace{\mathfrak{b}\gam_{n+r}}},\ldots, \gam_{n+m-1}|\Big)\\
=& \frac{1}{2} \sum_{n=1}^N  |[\mathfrak{a},\mathfrak{b}]\gam_n, \gam_{n+1}, \ldots , \gam_{n+m-1}|, 
\end{split} 
\]
which is non-zero since the center of $\mathfrak{sl}(m,\mathbb{R})$ is trivial.
\end{proof}

\subsection{The kernel of $\omega_2$}

We now show that the kernel of $\omega_2$, when evaluated on a pair of vector fields in $\V_N^1$, is $2$-dimensional. First recall that, given the vector field $X=\{X_n\}$, with 
\[
X_n = v_n^0\gamma_n+v_n^1\gamma_{n+1}+\dots+v_n^{m-1}\gamma_{n+m-1} = \rho_n \v_n,
\]
 the components of the associated  matrix $Q$, defining the frame evolution, are 
 \[
 Q_n = \begin{pmatrix}\v_n & K_n\v_{n+1}& K_n K_{n+1} \v_{n+2}& \dots& K_nK_{n+1}\dots K_{n+m-2}\v_{n+m-1}\end{pmatrix}.
\]
When $X\in \V_N^1$, the condition that $Q_n$ is traceless can be restated as 
\begin{equation}
\label{LAPonQ}
\e_1^T \v_n + \sum_{\ell=1}^{m-1} \e_{\ell+1}^T K_nK_{n+1}\dots K_{n+\ell-1}\v_{n+\ell} = 0.
\end{equation}

\begin{theorem} In arbitrary dimension $m$, the kernel of $\omega_2$ is generated by the  vector fields 
\begin{equation}\label{kernel2}
X_n^1 = \gamma_{n+1}+\alpha_n\gamma_n, \quad\quad X^2_n = \gamma_{n+2}+b_n \gamma_{n+1}+\beta_n \gamma_n ,
\end{equation}
where $\alpha_n$ and $\beta_n$ are uniquely determined from the arc-length preservation condition, and where $b_n$ is the unique solution of the equation
\begin{equation}\label{b}
b_n - b_{n+m} = a_{n+2}^{m-1}-a_{n}^{m-1}.
\end{equation}
\end{theorem}

\begin{proof} First, we verify that the vector fields $X^1, X^2$ described in~\eqref{kernel2} are elements of the kernel of $\omega_2$. For general vector fields $X, Y\in \V_N$, $X_n = \rho_n \v_n, Y_n = \rho_n \w_n$, we compute (see, e.g., proof of Theorem~\ref{omega2br})
\[
\omega_2(X, Y) = \frac{1}{2}\sum_{n=1}^{N}\sum_{r=1}^{m-1}v_{n}^r\e_1^TK_{n}\dots K_{n+r-1}\w_{n+r} - w_n^r\e_1^TK_n\dots K_{n+r-1}\v_{n+r}.
\]
Since $\e_1^TK_n\dots K_{n+r-1} \e_1 = \e_1^TK_n \e_r= \e_1^T\e_{r+1} = 0$ for $r=1,\dots,m-1$,  the above is independent of the first component of both $\v_n$ and $\w_n$. Thus, we can assume $v_n^0 = w_n^0 = 0$ for now, with no loss of generality.

Setting $v_n^1 = 1$ and $v_n^r = 0$ for $r=2, \dots, m-1$, we get
\[
\begin{split}
\omega_2(X,Y) & =\sum_{n=1}^{N}\Big(\e_1^TK_n\w_{n+1}-\sum_{r=1}^{m-1}w_n^r\e_1^TK_n\dots K_{n+r-1} \e_2\Big) \\
& = \sum_{n=1}^{N}\Big(\e_1^TK_n\w_{n+1}-\sum_{r=1}^{m-1}w_n^r\e_1^TK_n\e_{r+1} \Big)= \sum_{n=1}^{N}(\e_1^TK_n \w_{n+1} - \e_1^TK_n \w_n). \\
&= (-1)^{m-1}\sum_{n=0}^{N-1}(w_{n+1}^{m-1} - w_n^{m-1}) = 0.
\end{split}
\]
This shows that $X=\gam_{n+1}$ belongs to the kernel of $\omega_2$, and so does its reparametrization $X_1=\mathcal{P}X$, since  $\omega_2(X,Y)$ is independent of $v^0_n$.

Assume now $v_n^1 = b_n, v_n^2 = 1$, and $v_n^r = 0$, for $r = 3, \dots, m-1$ . Then 
\[
\begin{split}
&\omega_2(X,Y)  = \sum_{n=1}^{N} \Big[b_n\e_1^TK_n \w_{n+1}+\e_1^TK_nK_{n+1}\w_{n+2}-\sum_{r=1}^{m-1}\left(w_n^r\e_1^TK_n\dots K_{n+r-1}(\e_3+b_{n+r}\e_2)\right)\Big] \\
& =\sum_{n=1}^{N}\Big(b_n\e_1^TK_n \w_{n+1}-\sum_{r=1}^{m-1}w_n^rb_{n+r}\e_1^TK_n \e_{r+1}+\e_1^TK_nK_{n+1}\w_{n+2}-\sum_{r=1}^{m-1}w_n^r\e_1^TK_nK_{n+1}\e_{r+1}\Big) \\
& =\sum_{n=1}^{N}\big(b_n\e_1^TK_n \w_{n+1}-b_{n+m-1}\e_1^TK_n \w_n+\e_1^TK_nK_{n+1}\w_{n+2}-\e_1^TK_nK_{n+1}\w_n\big) \\
& =\sum_{n=1}^{N}\left[(-1)^{m-1}(b_nw^{m-1}_{n+1}-b_{n+m-1}w_n^{m-1})+\e_1^TK_nK_{n+1}\big((w_{n+2}^{m-2}-w_n^{m-2})\e_{m-1}+(w_{n+2}^{m-1}-w_n^{m-1})\e_m\big)\right] \\
& = \sum_{n=1}^{N}\big((-1)^{m-1}(b_{n-1}-b_{n+m-1})w_n^{m-1}+(-1)^{m-1}(w_{n+2}^{m-2}- w_n^{m-2})+(-1)^{m-1}(w_{n+2}^{m-1}-w_n^{m-1})a_{n+1}^{m-1}\big) \\
& =\sum_{n=1}^{N}\big((-1)^{m-1}(b_{n-1}-b_{n+m-1}+a_{n-1}^{m-1}-a_{n+1}^{m-1})w_n^{m-1}\big), 
\end{split}
\]
which vanishes for all $\w_n$ if and only if $b_{n-1}=b_{n+m-1}-a_{n-1}^{m-1}+a_{n+1}^{m-1}$.

In order to prove that $X^1$ and $X^2$ are the generators of the kernel of $\omega_2$,  we have recourse to the algebraic relation obtained in Theorem~\ref{omega2br}
\[
\omega_2(X^g, X^h) = \{g, h\}_0,
\]
with $X^g, X^h$ the vector fields associated with the Hamiltonians $g, h$, and study the kernel of the Poisson bracket
\[
\{g, h\}_0=\sum_{n=1}^{N} {\bf h}_n^T \Pc_{2n} {\bf g}_n,
\]
where ${\bf g}_n,{\bf h}_n$ are invariant $(m-1)$-dimensional vectors and $\Pc_2$ is a skew-symmetric tensor.

An explicit expression for  $\Pc_2$ is given in~\eqref{P30} for dimension $m=3$, and can be found in reference~\cite{MW} for arbitrary $m$, written in terms of alternative, but equivalent coordinates (here denoted with $\k$):
\begin{equation}\tiny
\label{P2m}
\begin{pmatrix}\T k_{n}^2- k_n^2\T^{-1} & \dots & \T^{m-3} k_{n}^{m-2}-k_n^{m-2}\T^{-1} & \T^{m-2} k_{n}^{m-1} - k_n^{m-1}\T^{-1} & \T^{m-1} - \T^{-1}\\\\
\T k^3_{n}-k_n^3\T^{-2} & \dots& \T^{m-3} k_{n}^{m-1}-k_n^{m-1}\T^{-2}&\T^{m-2}-\T^{-2}&0\\
\vdots&\vdots&\vdots&\vdots&\vdots\\ \T k_{n}^{m-2}-k_n^{m-2}\T^{-(m-3)}& \T^2 k_{n}^{m-1}-k_n^{m-1}\T^{-(m-3)}&\T^3-\T^{-(m-3)}&\dots& 0\\ \\
\T k_{n}^{m-1}-k_n^{m-1}\T^{-(m-2)}& \T^2-\T^{-(m-2)}&0&\dots& 0\\\\ \T-\T^{-(m-1)}&0&0&\dots&0.
\end{pmatrix}
\end{equation}
For our purposes,  it suffices to show that the kernel of $\Pc_2$ is 2-dimensional.

Looking at the terms along the diagonal of~\eqref{P2m}, if ${\bf g}_n = (g^1_n, \dots g^{m-1}_n)^T$ is in the kernel of $\Pc_{2n}$, then 
\[
g^1_{n+1} = g^1_{n-m+1}, \quad \forall n=1, \ldots, N, 
\]
which can be solved uniquely as $g^1_n = c^1$, $c^1$ a constant, since $m$ and $N$ are assumed to be coprime. Going  back to $\Pc_{2n}$, we see that $g_n^2$ must satisfy the relation
\begin{equation}\label{g2}
g^2_{n-m} = g^2_n+c^1(k_{n-1}^{m-1}-k_{n-2}^{m-1}),
\end{equation}
which admits a unique solution if, and only if $\sum_{n=1}^{N} (k_{n-1}^{m-1}-k_{n-2}^{m-1}) = 0$. The telescopic sum is indeed zero, because the $\k_n$'s are periodic.
Also, the next entry $g_n^3$ must satisfy
\[
g_{n-m+3}^3 - g^3_{n+3} = k_{n+2}^{m-1}g_{n+2}^2 - k_n^{m-1}g_{n-m+3}^2+c^1(k_{n+1}^{m-2}-k_n^{m-2}).
\]
These equation admits a unique solution if and only if
\[
\sum_{n=1}^{N} \Big(k_{n+2}^{m-1}g_{n+2}^2 - k_n^{m-1}g_{n-m+3}^2+c^1(k_{n+1}^{m-2}-k_n^{m-2})\Big) =\sum_{n=1}^{N} \big( k_{n+2}^{m-1}g_{n+2}^2 - k_n^{m-1}g_{n-m+3}^2\big)= 0.
\]
We will show that this equation has no solution unless $c^1 = 0$, which implies that $g_n^2 = c^2$ and $g_n^3$ is the solution of an equation of the same form as~\eqref{g2} with $c^1$ replaced with $c^2$.

In fact, we will prove a stronger result: assume
\begin{equation}\label{equk}
g_{n-m} = g_n + c(b_{n-1}-b_{n-\ell}), \quad \ell=1, \ldots, m-1,
\end{equation}
where $c$ is a constant. Then, generically,
\begin{equation}\label{equk2}
\sum_{n=1}^{N} b_{n} g_{n} - b_{n-\ell} g_{n-m+1}=0
\end{equation}
if and only if, $\ell = m-1$. 

This result implies that the only elements of the kernel are of the form ${\bf g}_n =(0, \dots, 0, 0, c)$ and ${\bf g}_n =(0, \dots, 0, h_n^{m-2}, c)$, where $c$ is constant and $g_n^{m-2}$ satisfies
\[
g^{m-2}_{n-m} = g^{m-2}_n+c(k_{n+1}^{m-1}-k_{n+m-1}^{m-1}).
\]
It follows that the kernel of $\omega_2$ is $2$-dimensional.

Let us proceed to prove this result. Given $b=\{b_n\}_{n=1}^N$, we write $\{b_{n-1}\}=Eb$ in terms of the $N$-dimensional permutation matrix
\[
\tiny
E = \begin{pmatrix}0&0& \dots&0&1\\ 1&0& \dots&0&0\\ 0& \ddots & \ddots &\vdots& \vdots\\ 
\vdots& \ddots &\ddots& 0&0\\0&\dots& 0 &1&0\end{pmatrix}.
\]
Then, formula~\eqref{equk} can be written as 
\[
(E^m-I) g = c(E-E^\ell) b.
\]
If $g_n$ are not all equal to the same constant, rewriting the above as
\[
 (E-I)(E^{m-1}+E^{m-2}+\dots +I) g = cE(I-E)(I+E+\dots+E^{\ell-2}) b,
 \]
 we can solve for $g$, since $E^{m-1}+E^{m-2}+\dots + I$ is invertible when $m$ and $N$ are coprime:
 \[
 g = -c(E^{m-1}+E^{m-2}+\dots +I)^{-1}E(I+E+\dots+E^{\ell-2}) b.
 \]
 Equation \eqref{equk2} becomes
\[
\begin{split}
0 & =\big((I-E^{\ell-m+1})b\big)^T g = b^T(I-E^{m-\ell-1})g\\
& =(-c) b^T(I-E^{m-\ell-1})(E^{m-1}+E^{m-2}+\dots +I)^{-1}E(I+E+\dots+E^{\ell-2}) b.
\end{split}
\]
 Since $b$ is arbitrary, we must have
\begin{equation}\label{equk3}
(I-E^{m-\ell-1})(E^{m-1}+E^{m-2}+\dots +I)^{-1}E(I+E+\dots+E^{\ell-2})=0.
\end{equation}
Permuting the first two terms (the $E^j$'s are circulant matrices and thus commute with one another) and multiplying by $E^{m-1}+E^{m-2}+\dots +I$ on the left, equation 
\eqref{equk3} reduces to 
\[
(I-E^{m-\ell-1})E(I+E+\dots+E^{\ell-2})=(E+E^2+\dots+E^{\ell-1})-(E^{m-\ell}+E^{m-\ell+1}+\dots+E^{m-2})=0.
\]
Since the $E^j$'s, $j=1,\dots,N$ are the circulant matrices generators, the equation above holds if and only if $\ell = m-1$.
\end{proof}

 As in the case of the kernel of $\omega_1$, the generators $X^1$ and $X^2$ of the kernel of $\omega_2$ have rather simple Hamiltonians with respect to $\omega_1$.

\begin{proposition}\label{HamforKern2VF}
The Hamiltonian functions associated to the vector fields $X^1$ and $X^2$ in (\ref{kernel2}) by means of $\omega_1$ are given by
\[
f_1(\gamma) = -\sum_{n=1}^N a_n^{m-1} =m \sum_{n=1}^N \alpha_n, \quad\quad f_2(\gamma) = \frac m2\sum_{n=1}^N \beta_n.
\]
\end{proposition}
\begin{proof}
We first find the associated matrices $Q^{X^1}$ and $Q^{X^2}$,  and use \eqref{deltaform} to find the $\omega_1$-Hamiltonians. (See also the Remark following~\eqref{diffsymp}.)

Since $X^1_n = \rho_n \r_n$, with $\r_n = (\alpha_n, 1,0, \dots, 0)^T$, we compute
\[
\begin{split}
K_n \r_{n+1} & = \begin{pmatrix}0 & \alpha_{n+1} & 1 & 0 & \dots & 0\end{pmatrix}^T, \\
& \vdots  \\
 K_nK_{n+1}\dots K_{n+m-3}\r_{n+m-2} & = \begin{pmatrix}0 & \dots & 0 &\alpha_{n+m-2} &1\end{pmatrix}^T, \quad \text{and} \\
 K_nK_{n+1}\dots K_{n+m-2}\r_{n+m-1} & = K_ne_m+\begin{pmatrix}0 & \dots & 0 &\alpha_{n+m-2}\end{pmatrix}^T.
\end{split}
\]
Then, from~\eqref{Q}, we get
\begin{equation}\label{QX1}
Q_n^{X^1} = K_n + \begin{pmatrix} \alpha_n&0&\dots&0\\0&\alpha_{n+1}&\dots&0\\\vdots&\ddots&\ddots&\vdots\\ 0&\dots&0&\alpha_{n+m-1}\end{pmatrix}.
\end{equation}
A longer, but similar calculation gives
\begin{equation}\label{QX2}
Q_n^{X^2} = \begin{pmatrix}\beta_n&0&\dots&0\\0&\beta_{n+1}&\dots&0\\\vdots&\ddots&\ddots&\vdots\\ 0&\dots&0&\beta_{n+m-1}\end{pmatrix}+ K_n \begin{pmatrix}b_n&0&\dots&0\\0&b_{n+1}&\dots&0\\\vdots&\ddots&\ddots&\vdots\\ 0&\dots&0&b_{n+m-1}\end{pmatrix}+ K_n K_{n+1}.
\end{equation}
Requiring $\text{tr}(Q_n^{X^j})=0, j=1,2$, determines $\alpha_n$ and $\beta_n$: 
\begin{equation}\label{abexplicit}
\alpha_n = -\Rc^{-1}a_n^{m-1}, \quad\quad \beta_n = -\Rc^{-1}\big(a_n^{m-2}+a_{n+1}^{m-2}+a_{n}^{m-1}(b_{n+m-1}+a_{n+1}^{m-1})\big),
\end{equation}
where  $\Rc=1+\T+ \ldots +\T^{m-1}$.
Matching  $\g_1$-components of~\eqref{QX1} and ~\eqref{QX2} with the right-hand side of~\eqref{deltaform}, gives the gradients of $f_1$ and $f_2$ (assuming such Hamiltonians exist):
\[
\nabla_n f_1 = (0,0,\dots,0,-1)^T, \quad\quad \nabla_n f_2 = (0,0,\dots,0,-1,-b_{n+m-1}-a_{n+1}^{m-1})^T.
\]
The $f_1$ chosen above has the correct gradient, since
\[
f_1(\gamma) = -\sum_{n=1}^N a_n^{m-1} = \sum_{n=1}^N \Rc \alpha_n = \sum_{n=1}^N (\alpha_n+\alpha_{n+1}+\dots+\alpha_{n+m-1}) = m\sum_{n=1}^N \alpha_n.
\]
To justify the choice for $f_2$, we need more work. First,~\eqref{QX2} suggests an $f_2$ of the form
\[
f_2 = -\sum_{n=1}^N a_n^{m-2} - g(\a), \quad\text{with}~~ \frac{\partial g}{\partial a_n^{m-1}} = b_{n+m-1}+a_{n+1}^{m-1}.
\]
Equation \eqref{b}, determing $b_n$ up to a constant, can be rewritten as (shifting indices)
 \[
 \Rc (b_n- b_{n-1}) = a_{n-1}^{m-1}- a_{n+1}^{m-1} = (1-\T)(1+\T) a_{n-1}^{m-1}.
 \]
Since neither $b_n$ nor $a_n^{m-1}$ are in the kernel of $1-\T$, substituting $\Rc (b_n- b_{n-1}) =-(1-\T)b_{n-1}$, we arrive at
 \begin{equation}\label{bexplicit}
 b_{n} = -\Rc^{-1}(a_{n}^{m-1} + a_{n+1}^{m-1}) = \alpha_n+\alpha_{n+1},
 \end{equation}
which gives
\[
 b_{n+m-1}+a_{n+1}^{m-1} = -\Rc^{-1}(a_{n+m-1}^{m-1}+a_{n+m}^{m-1}) + a_{n+1}^{m-1} = -\Rc^{-1}(a_{n+1}^{m-1}+a_{n+2}^{m-1}+\dots+a_{n+m-2}^{m-1}).
 \]
Let $c_n = \Rc^{-1} a_n^{m-1}$ and compute
\[
\frac{\partial g}{\partial c_n}   = \sum_{r=1}^{N} \frac{\partial g}{\partial a_r^{m-1}}\frac{\partial a_r^{m-1}}{\partial c_n}=\sum_{r=n-(m-1)}^{n} \frac{\partial g}{\partial a_r^{m-1}},
\]
where we used $a_n^{m-1} = \Rc c_n = c_n+c_{n+1}+\dots+c_{n+m-1}$. From ${\partial g}/{\partial a_n^{m-1}} = b_{n+m-1}+a_{n+1}^{m-1}$, we then obtain
\[
\begin{split}
\frac{\partial g}{\partial c_n} & =  \sum_{r=n-(m-1)}^{n}\Rc^{-1}\big(a_{r+1}^{m-1}+a_{r+2}^{m-1}+\dots+a_{r+m-2}^{m-1}\big)=  \sum_{r=n-(m-1)}^{n}\big(c_{n+1}^{m-1}+c_{n+2}^{m-1}+\dots+c_{n+m-2}^{m-1}\big)\\
&=(m-2)\big(c_{n+1}+c_n+c_{n-1}\big)+(m-3)\big(c_{n+2}+c_{n-2}\big)+ \dots\\
& \dots\, +2\big(c_{n+m-3}+c_{n-(m-3)}\big)+c_{n+m-2}+c_{n-(m-2)},
\end{split}
\]
which is the gradient of the function
\[
g(\a) = \sum_{n=1}^N \frac {m-2}2 c_n^2 + c_n \left((m-2)c_{n+1}+(m-3) c_{n+2}+\dots+2c_{n+m-3}+c_{n+m-2}\right).
\]
However, $g(\a)=\tfrac{1}{2} \sum_{n=1}^N a_n^{m-1}(b_{n+m-1}+a_{n+1}^{m-1})$, in fact
\[
\begin{split}
&\sum_{n=1}^N a_n^{m-1}\big(b_{n+m-1}+a_{n+1}^{m-1}\big)= \sum_{n=1}^N a_n^{m-1}\Rc^{-1}\big(a_{n+1}^{m-1}+a_{n+2}^{m-1}+\dots+a_{n+m-2}^{m-1}\big)\\ &= \sum_{n=1}^N\big(c_n+c_{n+1}+\dots+c_{n+m-1}\big)\big(c_{n+1}+c_{n+2}+\dots+c_{n+m-2}\big)\\ 
&=\sum_{n=1}^N (m-2)c_n^2 + 2c_n\big((m-2)c_{n+1}+(m-3)c_{n+2}+\dots+2c_{n+m-3}+c_{n+m-2}\big).
\end{split}
\]
Combining this with
\[
\frac{m}{2}\sum_{n=1}^N \beta_n=\frac{1}{2}\sum_{n=1}^N \Rc \beta_n =  -\sum_{n=1}^N a_n^{m-2} - \frac{1}{2}\sum_{n=1}^N a_n^{m-1}(b_{n+m-1}+a_{n+1}^{m-1}),
\]
completes the verification for $f_2$.
\end{proof}

The generators of the kernel of $\omega_2$ are commuting vector fields with respect to the symplectic structure defined on the moduli space of $\M^1_N$ by $\omega_1$.  
\begin{theorem}
\label{X12commute}
Let $X^1$ and $X^2$ be as in (\ref{kernel2}). Then
\[
\omega_1(X^1,X^2) = 0.
\]
\end{theorem}
\begin{proof}
We only need to show that $X^1$ preserves the $\omega_1$-Hamiltonian associated with $X^2$, i.e.
\[
\frac m2\sum_{n=1}^N X^1(\beta_n) = 0.
\]
Using \eqref{abexplicit}) and \eqref{bexplicit}, the above is equivalent to verifying that 
\[
\sum_{n=1}^N X^1\big(a_n^{m-2}+a_{n+1}^{m-2}+a_{n}^{m-1}(b_{n+m-1}+a_{n+1}^{m-1}) \big)
= \sum_{n=1}^N X^1\big(2a_n^{m-2}-a_{n}^{m-1}(\alpha_{n+1}+\dots+\alpha_{n+m-2})\big)=0.
\]
Substituting $X^1(\gamma_{n+r})=\gamma_{n+r+1}+\alpha_{n+r}\gamma_{n+r}$ from equation~\eqref{kernel2}, 

$a^r_n=|\gam_n,\dots,\gamma_{n+r-1}, \gamma_{n+m}, \gamma_{n+r+1},\dots,\gamma_{n+m-1}|$ from equation~\eqref{invariantsa}, and  $\gamma_{n+m+1} = \sum_{s=1}^{m-1} a_{n+1}^{s}\gamma_{n+s+1}+(-1)^{m-1}\gamma_{n+1}$, we compute
\[
\begin{split}
X^1(a_n^{m-1}) & = \sum_{r=0}^{m-2}\big( |\gamma_{n}, \dots,\gamma_{n+r-1}, \gamma_{n+r+1}+\alpha_{n+r}\gamma_{n+r},\gamma_{n+r+1},\dots,\gamma_{n+m-2},\gamma_{n+m}| \\
& +|\gamma_n,\dots,\gamma_{n+m-2},\gamma_{n+m+1}+\alpha_m\gamma_{n+m}|\big)\\
& =a_n^{m-1}\sum_{r=0}^{m-2}\big(\alpha_{n+r}+a_n^{m-1}\alpha_{n+m} -a_{n}^{m-2}+a_{n+1}^{m-2}+a_{n+1}^{m-1}a_n^{m-1}\big).
\end{split}
\]
Since $\sum_{r=1}^m\alpha_{n+r} = -a_{n+1}^{m-1}$, equivalently $\alpha_n=-\Rc^{-1}a_n^{m-1}$, the expression can be further simplified to
\[
X^1(a_n^{m-1})= a_n^{m-1}(\alpha_n-\alpha_{n+m-1}) + a_{n+1}^{m-2}-a_n^{m-2},
\]
which also gives $X^1(\alpha_n)=-\Rc^{-1}a_n^{m-1}$.
A similar calculation leads to
\[
X^1(a_n^{m-2}) = a_n^{m-2}(\alpha_n-\alpha_{n+m-2}) +a_{n+1}^{m-3}-a_n^{m-3}.
\]
Combining the various identities, we get
\[
\begin{split}
&\sum_{n=1}^N X^1\big(2a_n^{m-2}-a_{n}^{m-1}(\alpha_{n+1}+\dots+\alpha_{n+m-2})\big)  = \sum_{n=1}^N\Big[ 2a_n^{m-2}(\alpha_n-\alpha_{n+m-2})\\
&-a_n^{m-1}(\alpha_n-\alpha_{n+m-1})(\alpha_{n+1}+\dots+\alpha_{n+m-2}) -( a_{n+1}^{m-2}-a_n^{m-2})(\alpha_{n+1}+\dots+\alpha_{n+m-2})
\\
&+a_{n}^{m-1}\Rc^{-1}\Big((\alpha_{n+1}-\alpha_{n+m})a_{n+1}^{m-1}+ \dots +(\alpha_{n+m-2}-\alpha_{n+2m-3})a_{n+m-2}^{m-1} 
\\
&+ a_{n+1}^{m-2}-a_n^{m-2}+ a_{n+2}^{m-2}-a_{n+1}^{m-2}+\dots+a_{n+m-2}^{m-2}-a_{n+m-3}^{m-1}\Big)\Big],
\end{split}
\]
where the last telescopic sum simplifies to
\[
a_{n+1}^{m-2}-a_n^{m-2}+ a_{n+2}^{m-2}-a_{n+1}^{m-2}+\dots+a_{n+m-2}^{m-2}-a_{n+m-3}^{m-1} = a_{n+m-2}^{m-2}-a_n^{m-2},
\]
 and 
\[
\sum_{n=1}^N ( a_{n+1}^{m-2}-a_n^{m-2})(\alpha_{n+1}+\dots+\alpha_{n+m-2})) = \sum_{n=1}^Na_n^{m-2}(\alpha_n-\alpha_{n+m-2}).
\]
We briefly outline the rest of the proof, omitting the details for sake of brevity. One need to show that: 1. the terms involving $a_n^{m-2}$ vanish; this requires careful accounting, aided by introducing 
$\Rc^\ast = \T^{-m+1}\Rc$ such that $(\Rc^{-1})^\ast = \Rc^{-1}\T^{m-1}$.  2. the 
terms involving $a_n^{m-1}$ also vanish; the procedure is lengthier (but still fairly straightforward) and makes use of $a_n^{m-1}=-\sum_{r=1}^m\alpha_{n+r} =$ to  show that
\[
\sum_n a_n^{m-1}(\alpha_n - \alpha_{n+m-1})(\alpha_{n+1}+\alpha_{n+2}+\dots+\alpha_{n+m-2}) = 0.
\]
 \end{proof}
 
\section*{Discussion and Open Questions}
By lifting the Hamiltonian structures from the space of invariant curvatures to the space of arc length parametrized polygons, not only we prove that the Poisson brackets derived in~\cite{MW} form a bi-hamiltonian pair on the space of geometric invariants, but we also provide a framework for studying the integrability of the geometric flows that reveals several interesting features.

The 2-form $\omega_1$ is symplectic on the moduli space of $\M^1_N$, since its kernel is generated by the infinitesimal rigid motions, which preserve the monodromy class. Thus $\omega_1$ provides, through relation~\eqref{diffsymp}, a natural correspondence between Hamiltonians and Hamiltonian vector fields. On the other hand, the 2-form $\omega_2$ has a non-trivial kernel, so it remains an open question whether a hierarchy of commuting vector fields can be constructed by means of a Lenard-Magri type of scheme. 

In two dimensions, the 2-form $\omega_2$ is a straightforward discretization of Pinkall's symplectic form on the space of star-shaped planar curves~\cite{P95}. For $m=3$, $\omega_2$ can also be regarded as a straightforward generalization of  the 2-form introduced in~\cite{CIM} for curves in centro-affine $\mathbb{R}^3$. Its expression is very simple in arbitrary dimension (see also Proposition~\ref{omega2alt}) and, for $m\ge 3$, its kernel  is generated by two interesting vector fields: $X^1$, a natural discretization of the translation flow (i.e.~the discrete counterpart of $\gamma_x$ in the continuous case), and $X^2$, which could be interpreted as a discretization of the second derivative. Both $X^1$ and $X^2$ have been reparametrized to make them arc length-preserving. Moreover, we find that there are no negative flows generated by the infinitesimal symmetries (Corollary~\ref{nointegrable}), and that $X^1$ and $X^2$ commute with respect to the $\omega_1$-Hamiltonian structure. This parallels the case of the integrable curve evolutions in centro-affine $\mathbb{R}^3$ discussed in ~\cite{CIM}, where geometric realizations of the Boussinesq hierarchy are constructed as a double hierarchy of commuting vector fields whose headers are the elements of the kernel of the pre-symplectic form. The fact that, also in the discrete case and in arbitrary dimension (not just for $m=3$), the integrable flows are generated by a pair of vector fields, is intriguing and suggests that $\omega_2$ and its continuous counterparts should be further investigated. 

We conclude with some remarks suggesting a possible way of constructing  integrable hierarchies, currently under investigation.

Define $\Vh_N^1$ to be the $\omega_1$-complement of the kernel of $\omega_2$, that is
\[
\Vh_N^1 = \left\{X \in \V_N^1 \ | \  \omega_1(X, X^1) = \omega_1(X, X^2) = 0\right\}.
\]
Theorem~\ref{X12commute} implies that $X^1, X^2 \in \Vh_N^1$.  Also, from a standard linear algebra argument, we have
\begin{lem}
\label{startrec}
If $X \in \Vh_N^1$, then there exists $Y\in \V_N^1$ such that 
\[
\omega_1(X, Z) = \omega_2(Y, Z), \qquad \forall Z\in \V_N^1.
\]
\end{lem}
Since  $\omega_1$ is symplectic when restricted to the leaves of the moduli space of $\M^1_N$, this suggests the following conjecture.
\begin{conj}
\label{hier} Let $X_1, X_2 \in \Vh_N^1$ and define $X_{k+2}\in \V_N^1$ as 
\[
\omega_1(X_k, Z) = \omega_2(X_{k+2}, Z).
\]
Then $X_k \in \Vh_N^1$, for all $k=1,2,3,\dots$. 
\end{conj}
Note that there are non-vacuous choices for the first two vector fields of the hierarchy generated by the conjecture. Choosing $X_1=X^1$ and $X_2=X^2$, the existence of $X_3$ and $X_4$ is guaranteed by Lemma~\ref{startrec}, and $X_1, X_2 \in \Vh_N^1$ since $\omega_1(X_1, X_2)=0$ by Theorem~\ref{X12commute}.  The next two vector fields $X_3, X_4$ induce evolutions on the geometric invariants as outlined in Corollary~\ref{a-p}. If Conjecture~\ref{hier} holds, such evolutions are guaranteed to be Hamiltonian because the Poisson brackets are compatible and $\{\ ,\, \}_R$ is symplectic; in other words, the $\g_1$-components $\p_3, \p_4$ of the frame evolutions $Q^{X_3}, Q^{X_4}$ must be the gradients of (not necessarily unique) functions $f_3, f_4$ generated by one step of a Lenard-type scheme. It follows from the main Theorem~\ref{compatibility} that the Hamiltonians $f_1, f_2$ associated with $X_1, X_2$\textemdash explicitly constructed in Proposition~\ref{HamforKern2VF}\textemdash commute with respect to the reduced bracket $\{\ ,\, \}_R$.  It is immediate to show that $f_3, f_4$ commute with each of $f_1$ and $f_2$ (e.g.~$\{f_3 ,f_2\}_R=\omega_1(X_3, X_2)=0$, since $X_3 \in \Vh_N^1$), and with each other. Indeed $\{f_3, f_4\}_R = \omega_1(X_3, X_4) = \omega_2(X_5, X_4) = \omega_1(X_5, X_2) = 0$ since the conjecture would ensure that $X_5 \in \Vh_N^1$. The recursive process would lead to a hierarchy on commuting geometric vector fields and associated hierarchy of commuting hamiltonian flows on the space on geometric invariants.  We remark that the latter differs from the hierarchy of flows conjectured in \cite{MW}. The study of this potentially new hierarchy is in progress.

\section*{Acknowledgements} We wish to thank Anton Izosimov for helpful suggestions. We gratefully acknowledge the support of the National Science Foundation:\, Gloria Mar\'i Beffa through grant DMS-0804541, and Annalisa Calini through grant DMS-1109017.

\clearpage

\section*{Corrigenda: Integrable Evolutions of Twisted Polygons in Centro-Affine $\mathbb{R}^m$}

In the paper \emph{Integrable Evolutions of Twisted Polygons in Centro-Affine $\mathbb{R}^m$}~\cite{CM}, the statement of Theorem 5.5  is incorrect and its proof flawed. While the vector fields $X^1$ and $X^2$ are in the kernel of $\omega_1$, they do not generate it. Fortunately, this theorem is at the end of the article and is not used in other parts of the work; so no other result is affected by it. 

The correct statement, together with its proof, appears in Theorem 7.2 of \cite{MB}, where the kernel is described in terms of fractional powers of difference operators. Using this different approach, it is shown that there are $m-1$ independent vector fields in the kernel of $\omega_1$, making the dimension of the kernel at least $m-1$.

\end{document}